\documentclass[12pt,a4paper]{article}

\usepackage{tabularx}
\usepackage{comment}
\usepackage[T1]{fontenc}
\usepackage[utf8]{inputenc}
\usepackage{lmodern}
\usepackage[a4paper, top=2.5cm, bottom=2.5cm,
            left=2.5cm, right=2.5cm]{geometry}
\usepackage{amsmath, amssymb, amsthm}
\usepackage{etoolbox}
\AtEndEnvironment{remark}{\hfill$\lozenge$}
\usepackage{mathtools}
\usepackage{bbm}
\usepackage{bm}
\usepackage{microtype}
\usepackage{setspace}
\usepackage{booktabs}
\usepackage{xcolor}

\usepackage{graphicx}
\usepackage[colorlinks=true,
            linkcolor=blue!60!black,
            citecolor=blue!60!black,
            urlcolor=blue!60!black]{hyperref}
\usepackage{natbib}
\usepackage{parskip}

\bibpunct{(}{)}{;}{a}{,}{,}

\graphicspath{{./}}

\theoremstyle{plain}
\newtheorem{theorem}{Theorem}[section]
\newtheorem{lemma}[theorem]{Lemma}
\newtheorem{corollary}[theorem]{Corollary}

\theoremstyle{definition}
\newtheorem{assumption}{Assumption}

\theoremstyle{remark}
\newtheorem{remark}{Remark}[section]

\newcommand{\E}{\mathbb{E}}
\newcommand{\R}{\mathbb{R}}
\newcommand{\calY}{\mathcal{Y}}
\newcommand{\calX}{\mathcal{X}}

\newcommand{\linf}{\ell^{\infty}}

\newcommand{\G}{\mathbb{G}}

\renewcommand{\H}{\mathbb{H}}
\newcommand{\wto}{\rightsquigarrow}
\newcommand{\dto}{\xrightarrow{d}}
\newcommand{\pto}{\xrightarrow{p}}
\newcommand{\thetazerot}[1]{\theta^{0}_{1,#1}}
\newcommand{\hatthetazerot}[1]{\hat{\theta}^{0}_{1,#1}}

\newcommand{\Thetamat}{\boldsymbol{\Theta}}
\DeclareMathOperator{\Cov}{Cov}
\DeclareMathOperator*{\argmin}{arg\,min}
\DeclareMathOperator*{\argmax}{arg\,max}
\DeclareMathOperator{\vecop}{vec}
\title{\large\bfseries
  A Synthetic Control Approach to Conditional Distributional Treatment Effects}

\author{Dominik Wied%
  \thanks{University of Cologne, Albertus-Magnus-Platz,
    50923~Cologne, Germany. E-mail: \texttt{dwied@uni-koeln.de}.}}

\date{\today}

\begin{document}
%% ===============================================================

\maketitle
\thispagestyle{empty}

%% ---------------------------------------------------------------
\begin{abstract}
\begin{singlespace}
\noindent
This paper proposes a synthetic control (SC) framework for the estimation of conditional distributional treatment effects. Identification rests on a parallel trends condition formulated in the parameter space of the semiparametric distribution regression (DR) model, which keeps the counterfactual conditional distribution within the model class. The weights solve a least-squares problem subject to an adding-up constraint, yielding a closed-form estimator. We derive the asymptotic distribution of the counterfactual estimator, with DR estimation error and weight estimation error contributing at the same rate to the asymptotic variance. Moreover, we propose a supremum test for the null of no treatment effect, whose limit is the supremum of a Gaussian process. Simulations illustrate that conditioning on covariates can reveal effects being difficult to detect from the unconditional distribution alone. An application to the 1992 New Jersey minimum wage increase using CPS data finds effects concentrated in the minimum-wage corridor for low-education, low-experience workers.
\medskip\noindent

\textbf{Keywords:} Causal inference, functional delta method, distribution regression, minimum wage, parallel trends.
\medskip\noindent

\textbf{JEL Codes:} C12, C14, C21, C25, J31, J38.
\end{singlespace}
\end{abstract}

\newpage
\setcounter{page}{1}

%% ---------------------------------------------------------------
\section{Introduction}
\label{sec:intro}
%% ---------------------------------------------------------------

The synthetic control~(SC) method, introduced by \citet{AbadieGardeazabal2003} and
formalised in \citet{AbadieEtAl2010}, has become one of the most widely used tools
for the evaluation of policy interventions in settings with a small number of treated
units; \citet{Abadie2021} provides a comprehensive review. The classical estimator constructs a weighted average of control unit outcomes
that approximates the pre-treatment trajectory of the treated unit, and uses this
synthetic unit as the counterfactual in the post-treatment period. \citet{DoudchenkoImbens2016} relax the non-negativity constraint on weights, allowing the synthetic control to lie outside the convex hull of the donor pool, and show that this flexibility improves pre-treatment fit. \citet{BenMichaelEtAl2021} propose an augmented SC estimator that combines outcome-model bias correction with SC weighting. \citet{ChernozhukovEtAl2021} develop an exact, finite-sample valid inference approach for SC based on conformal prediction. \citet{FermanPinto2021} study the properties of SC when pre-treatment fit is imperfect. \citet{ArkhangelskEtAl2021} develop a synthetic difference-in-differences estimator that accommodates heterogeneous treatment effects. \citet{Chen2023} establishes a connection between synthetic control and online learning, showing that the SC estimator can be interpreted as an instance of follow-the-leader, which yields oracle inequalities for counterfactual predictions even under adversarial outcomes. While powerful for estimating average treatment effects at the aggregate level, all these methods provide no information about distributional consequences, a critical limitation for policy questions involving inequality, minimum wages, or tax reforms.

The broader literature on distributional treatment effects includes the counterfactual wage decompositions of \citet{MachadoMata2005} and \citet{Melly2005} and \citet{ChernozhFVMelly2013}, the unconditional quantile regressions of \citet{FirpoFortinLemieux2009}, and the distributional DiD approaches of \citet{ChernozhFVMelly2013} and \citet{Fernandezval2026}. None of these papers considers a synthetic control approach.

A first important extension was proposed by \citet{Gunsilius2023}, who replaces scalar
outcomes with quantile functions, thereby allowing inference on heterogeneous treatment
effects across the distribution. His estimator matches unconditional quantile functions
of the control units to that of the treatment unit in the pre-treatment period, using
a constrained quantile-on-quantile regression. While this distributional synthetic control
(DSC) method is a substantial advance, it is inherently unconditional: It cannot
accommodate individual-level covariates such as education, experience or demographics
that are central to labour economics and related fields. 

The present paper fills this gap by proposing a SC estimator for \emph{conditional}
distribution functions. \citet{ChenFeng2026} extend the DSC framework to group-level heterogeneity, but their heterogeneity is driven by unobservable group membership rather than observable covariates, leaving the gap addressed here open. Conditioning on covariates is not simply a refinement of the unconditional approach: it changes the identification problem in a fundamental way. When the outcome distribution depends on individual characteristics, as wages depend on education and experience, the counterfactual object of interest is a conditional distribution function, and any identification assumption must be formulated at the level of the model that generates it.

We work within the semiparametric DR framework of
\citet{ForesiPeracchi1995}, where the conditional distribution function of an outcome $Y$
given covariates $X \in \R^p$ is modeled as $F(y\mid x)=\Lambda(x'\theta(y))$ for a known link
$\Lambda$ and unknown parameter function $\theta:\calY\to\R^p$. This model has appealing
properties in applications: it handles mass points, irregularities, and multimodality
naturally, and estimation reduces to a sequence of binary outcome regressions indexed by $y$. The DR literature includes \citet{KneibEtAl2023} and \citet{Klein2024} on flexible extensions, \citet{spady:2025} on semiparametric efficiency, and \citet{Wied2024} on instrumental variable estimation in the DR framework; \citet{RotheWied2013} derive the asymptotic theory used here; and \citet{BiewenErhardt2025} apply distribution regression to minimum wage analysis. To the best of our knowledge, the combination of synthetic control methods with semiparametric distribution regression has not been studied in the existing literature.

Our identification assumption \emph{Parallel Trends in Parameters}~(PTP) is
formulated directly in the DR parameter space: the counterfactual parameter function of the
treated group evolves from the pre-treatment period to the post-treatment period along the same
weighted combination of trends observed in the control groups, requiring the weights to sum to
one. The economic rationale is that the DR parameters (e.g. quantile-specific returns to education and experience in a Mincer framework) are driven by a small number of common aggregate shocks (e.g. skill-biased technological change, business-cycle fluctuations, globalisation) whose differential local impact is captured by state-specific factor loadings. PTP formalises the condition that the factor loadings of a treated unit can be approximated by a weighted combination of donor states' loadings, which is the direct functional analogue of the latent-factor rationale underlying classical SC methods \citep{AbadieEtAl2010}. This approximation is natural in the DR \emph{parameter} space rather than at the CDF level, because the parameter function is linear in the common factors whereas the CDF is not. This formulation has also two other important advantages over a parallel trends condition on the CDF level.
First, it ensures that the counterfactual conditional distribution function remains within the DR
model class, so that all estimators have a direct parametric interpretation. Second, it allows us
to formulate the weight estimation problem as a finite-dimensional quadratic program (QP) at each
value of $y$, or equivalently as a single infinite-dimensional QP over the function space
$\linf(\calY)^p$. Dropping the traditional non-negativity constraint yields a closed-form weight
estimator, avoids interior-solution requirements, and transparently accounts for potentially
negative weights that arise when some control units lie outside the convex hull of the treated unit.

One contribution of the paper is the asymptotic analysis of the proposed estimator for the case of fixed number of groups $J$ and increasing individual observations within each group. This framework fits to micro-data applications such as the Current Population Survey (CPS) for US states, which we consider, as the number of states is small relative to the within-group sample size $n$. Statistical precision derives
from this sample size, not from the number of time periods $T$. In fact, one pre-period is sufficient for obtaining consistent parameter estimators, whereas more pre-periods increase the precision of the weight estimators. An important result is the joint weak convergence of the pointwise counterfactual CDF differences to a Gaussian vector, where both DR estimation error and weight estimation error contribute at the same $\sqrt{n}$ rate. Based on this, we show asymptotic normality of the integrated squared treatment effect estimator when the true effect is non-zero.
For testing the null hypothesis of no treatment effect, we propose a supremum statistic over the pointwise CDF differences, applicable either on the whole support of $Y$ or a subset. Its null distribution is the supremum of a mean-zero Gaussian process. This supremum statistic has a non-degenerate null distribution, which is approximated by a Gaussian process simulation, and power growing to one with $n$. A plug-in confidence interval for the integrated squared effect is reported in a second step, conditional on rejection. We characterize under which types of deviation from PTP the results still hold.

If multiple pre-periods are available, it is possible to perform pre-trend diagnostics. We use them to perform causal inference in an empirical application to the 1992 New Jersey minimum wage increase. In fact, we see that the DR-SC approach reveals a clear heterogeneity pattern that is invisible in aggregate analyses. The minimum wage effect is sharply concentrated in
the corresponding corridor for low-education and low-experience workers, the group most directly affected by the
policy. For other groups the effects are at most marginal, in particular for high-education/experience workers.

The paper proceeds as follows. Section~2 introduces the framework for general $T_0\geq 1$ pre-treatment and $T_1\geq 1$ post-treatment periods and discusses point as well as confidence estimation of the treatment effect. Moreover, the supremum test and the pre-trend test for the case $T_0\geq 2$ are presented. Section~3 derives the supporting asymptotic theory. Section~4 presents simulation evidence, which highlights the advantages of the conditional over an unconditional approach. Section~5 contains the empirical application, Section~6 concludes.

\section{Framework and Estimation}
\label{sec:framework}
%% ---------------------------------------------------------------
\subsection{Framework}
Consider $J+1$ groups $i=1,\ldots,J+1$ at $T=T_0+T_1$ time points, where $t=1,\ldots,T_0$ are pre-treatment periods and $t=T_0+1,\ldots,T_0+T_1$ are post-treatment periods. Group $i=1$ is
the treatment group; groups $i=2,\ldots,J+1$ form the donor pool of potential control groups.
In all pre-treatment periods none of the groups is treated. Only group $i=1$ receives the treatment from period $T_0+1$ onward.
Within each group-time cell $(i,t)$ we observe an i.i.d.\ sample of $n_{it}$ individuals with outcome
$Y\in\calY\subseteq\R$ and a $p$-dimensional covariate vector $X\in\calX\subseteq\R^p$. Typical applications
in labour economics have $J$ small (e.g., $J+1=10$ to 50 federal states or countries) and $n_{it}$ large
(e.g., several thousand individuals per group-time cell). We fix $J$ and let $n:=\sum_{i,t}n_{it}\to\infty$,
assuming $n_{it}/n\to r_{it}\in(0,\infty)$.

For each group $i$ and time period $t$, the conditional distribution function of $Y$ given $X=x$ satisfies
\begin{equation}
  \label{eq:DR}
  F_{it}(y\mid x) = \Lambda\bigl(x'\theta_{it}(y)\bigr),
  \quad y\in\calY,\;x\in\calX,
\end{equation}
where $\Lambda:\R\to(0,1)$ is a known strictly increasing link function (such as the standard normal CDF
for the probit link, or the logistic CDF for the logit link) with derivative $\lambda$ and $\theta_{it}:\calY\to\R^p$ is an unknown
c\`adl\`ag parameter function. For fixed $y$, $\theta_{it}(y)$ is the coefficient vector in the binary outcome
regression of $\mathbf{1}\{Y\leq y\}$ on $X$. Numerical sensitivity analyses in \citet{dette2025} confirm that the choice of the link function has a negligible impact on estimation results.

\begin{remark}[Monotonicity]
\label{rem:monotonicity}
The model ensures that $y\mapsto F_{it}(y\mid x)$ is a valid distribution function (monotone in $y$, taking
values in $(0, 1)$) if $y\mapsto x'\theta_{it}(y)$ is non-decreasing for all $x$ in the support. This
monotonicity can be imposed ex post via rearrangement \citep{Chernozhukov2010} or isotonic regression \citep{Wied2024}.
\end{remark}

We are interested in the treatment effect on a region of the outcome distribution $\calY_0\subseteq\calY$, measured via the integrated squared
discrepancy between the observed and counterfactual conditional distributions:
\begin{equation}
  \label{eq:effect}
  f_t(x,\calY_0) := \int_{\calY_0}
  \bigl[F_{1,t}(y\mid x)-F^{0}_{1,t}(y\mid x)\bigr]^2\,dy
  =: \int_{\calY_0}\Delta_t(y\mid x)^2\,dy, \quad x\in\calX,
\end{equation}
where $F^{0}_{1,t}(y\mid x)$ is the counterfactual conditional distribution of the treated group in
post-treatment period $t$, and
\begin{equation}
  \label{eq:Delta}
  \Delta_t(y\mid x) := F_{1,t}(y\mid x) - F^{0}_{1,t}(y\mid x)
\end{equation}
is the pointwise CDF difference. The integrated
discrepancy $f_t(x,\calY_0)$ captures both the magnitude and the shape of the treatment effect on the conditional
distribution at covariate value $x$ over the region $\calY_0$. It equals zero if and only if $\Delta_t(y\mid x)=0$ for almost all $y\in\calY_0$.
Setting $\calY_0=\calY$ recovers the full integrated effect; a proper subset $\calY_0\subsetneq\calY$ focuses attention on a specific region of the distribution, such as the minimum-wage corridor $[\mathrm{MW}_{\mathrm{old}},\mathrm{MW}_{\mathrm{new}}+\varepsilon]$.
One may also consider the doubly integrated measure $\int_{\calX}f_t(x,\calY_0)dF_{X}(x)$, which averages over the
covariate distribution. If not denoted otherwise, $f_t(x,\calY_0) = f_t(x)$ for notational simplicity in the following.

The fundamental identification challenge is that $F_{11}^{0}(y|x)$ is not observed. We identify it via the following assumption, which has to hold on the set of interest $\calY_0$.

\begin{assumption}[Parallel Trends in Parameters, PTP]\label{ass:ptp}
There exist weights $w_2,\ldots,w_{J+1}\in\R$ with $\sum_{i=2}^{J+1}w_i=1$ such that for
all $y\in\calY_0$ and all post-treatment periods $t=T_0+1,\ldots,T_0+T_1$,
\[
  \thetazerot{t}(y)-\theta_{1,T_0}(y)
  =\sum_{i=2}^{J+1}w_i\bigl[\theta_{i,t}(y)-\theta_{i,T_0}(y)\bigr],
\]
where $\theta_{1,t}^{0}(y)$ is the DR parameter function of the counterfactual $F_{1,t}^{0}(y\mid \cdot)=\Lambda(\cdot'\theta_{1,t}^{0}(y))$ for each post-treatment period $t$.
\end{assumption}

Assumption~PTP states that the trend in the treated group's DR parameter function from the pre- to the
post-treatment period would have equalled a weighted combination of the control groups' trends, absent
treatment. This is a direct analogue of the classical parallel trends assumption in the difference-in-differences
(DiD) literature, formulated in the parameter space of the DR model rather than on the level of the outcome
or its conditional mean. We impose only the adding-up constraint; negative weights arise when some control
units lie outside the convex hull of the treated unit. Differently to the case of negative weights in regression-based DiD \citep{CallawayLi2019}, this is not a problem of model specification, but a particular feature, which yields flexibility. \citet{DoudchenkoImbens2016} make a similar argument in the SC context and show that allowing negative weights improves pre-treatment fit.\\

\begin{remark}[Economic motivation of PTP]\label{remark:economic}
Assumption~PTP can be motivated by a
latent factor model of the distribution regression parameters. Suppose
that, absent treatment, the conditional wage distribution in each state is shaped
by a small number of common macroeconomic factors such as inflation,
technological change, or business-cycle conditions. These forces affect all states, but with different intensities depending on each state’s industry composition, labour market structure, and demographic profile.

Formally, assume that the
untreated DR parameter function admits the representation
\[
\theta^{0}_{it}(y) = \mu_{i}(y) + \Lambda_{i}(y)' F_{t} + u_{it}(y),
\]
where $F_{t} \in \mathbb{R}^{r}$ is a vector of common factors, $\Lambda_{i}(y)$
is a vector of state-specific factor loadings, $\mu_{i}(y)$ captures
time-invariant heterogeneity, and $u_{it}(y)$ is an idiosyncratic disturbance. If
the factor loadings of the treated unit can be approximated by a weighted average of the loadings of the donor states,
$\Lambda_{1}(y) \approx \sum_{i=2}^{J+1} w_{i}\Lambda_{i}(y)$ with
$\sum_{i=2}^{J+1} w_{i} = 1$, then the counterfactual evolution of the treated unit
satisfies approximately
\[
\thetazerot{t}(y) - \theta_{1,T_0}(y) \approx
\sum_{i=2}^{J+1} w_{i}\bigl(\theta_{i,t}(y) - \theta_{i,T_0}(y)\bigr).
\]
Thus PTP is the direct functional analogue of the latent-factor rationale
underlying classical synthetic control methods, see \citet{Abadie2021}.

This latent-factor structure has empirical support in the labour economics literature. In a Mincer framework, the DR parameters $\theta_{it}(y)$ collect quantile-specific returns to education and experience. The time variation in these returns (documented across US states by \citet{KatzMurphy1992} and \citet{CardLemieux2001}) is well explained by a small number of common demand shocks: Skill-biased technological change raises returns to education uniformly across states but with different intensities depending on the local industry mix, aggregate business-cycle fluctuations compress or expand the lower tail of the wage distribution through their differential impact on low-education employment \citep{AutorManningSmith2016} and trade exposure affects manufacturing-intensive states more strongly \citep{AutorDornHanson2013}. Each of these forces maps naturally to a common factor $F_t$, with state-specific industry composition determining the loading $\Lambda_i(y)$.

An advantage of formulating PTP in the DR parameter space (rather than directly at the CDF level) is that the factor structure is \emph{linear} in $\theta_{it}(y)$, whereas the corresponding CDF $F_{it}(y|x)=\Lambda(x'\theta_{it}(y))$ is a nonlinear. This linearity delivers constant weights $w$ across the full distribution, which is economically natural (the same macroeconomic forces act on all quantiles of the conditional distribution).
\end{remark}

\begin{remark}[Pre-treatment balance]
\label{rem:balance}
A natural sufficient condition for PTP is that (i)~the pre-treatment parameter functions satisfy
$\theta_{1,t}(y)=\sum_{i=2}^{J+1}w_i\theta_{i,t}(y)$ for all $y$ and all $t=1,\ldots,T_0$ (perfect pre-treatment balance), and
(ii)~the control groups would have followed the same weighted trend as the treatment group in the absence
of treatment. Under perfect pre-treatment balance, Assumption~PTP simplifies to $\thetazerot{t}(y)=\sum_{i=2}^{J+1} w_i\theta_{i,t}(y)$ for each post-period $t$, which means the counterfactual DR parameter is simply a weighted average of the period-$t$ DR parameters of the control groups. We will use this simplification as the basis for estimation.
\end{remark}

\begin{remark}[Comparison to CDF-level parallel trends]
\label{rem:cdf_comparison}
One might alternatively formulate parallel trends directly on the CDF level:
$F_{1,t}^{0}(y\mid x)=\sum_{i=2}^{J+1}w_i F_{i,t}(y\mid x)$ for all $y,x$, and each post-period $t$. Under Assumption PTP and the
DR model, this holds approximately. A second-order Taylor expansion yields the discrepancy:
\[
  \Lambda\Bigl(x'\sum_{i=2}^{J+1}w_i\theta_{i,t}(y)\Bigr)-\sum_{i=2}^{J+1}w_i\Lambda\bigl(x'\theta_{i,t}(y)\bigr)
  \approx-\frac{1}{2}\lambda'\bigl(x'\overline{\theta}_t(y)\bigr)\cdot \mathrm{Var}_{w}\bigl(x'\theta_{i,t}(y)\bigr),
\]
where $\overline{\theta}_t(y)=\sum_i w_i\theta_{i,t}(y)$ and $\lambda=\Lambda'$. The approximation error is thus
proportional to the variance of the DR index $x'\theta_{i,t}(y)$ across control groups, small when control
groups are similar and large when they are heterogeneous.
\end{remark}

\subsection{Point Estimation}
For each cell $(i,t)$ and $y$ on a grid $\calY_m=\{y_1,\ldots,y_m\}\subset\calY_0$, the DR
parameter $\theta_{it}(y)$ is estimated by maximizing the log-likelihood of the binary outcome model. For a fixed $y\in\mathcal{Y}$, the estimator solves:
\begin{equation}
  \label{eq:DR_estim}
  \hat{\theta}_{it}(y)=\argmax_{\theta\in\R^{p}}\frac{1}{n_{it}}\sum_{k=1}^{n_{it}}\bigl[\mathbf{1}\{Y_{itk}\le y\}\log \Lambda(X_{itk}'\theta)+\mathbf{1}\{Y_{itk}>y\}\log(1-\Lambda(X_{itk}'\theta))\bigr].
\end{equation}
This is a standard probit or logit regression of $\mathbf{1}\{Y\le y\}$ on $X$, estimated separately for each $y$ on a finite grid $\mathcal{Y}_{m}=\{y_{1},...,y_{m}\}\subset\mathcal{Y}_0$ (e.g. the 0.05,0.1,\ldots quantiles of the distribution of $Y$).

The weights $w=(w_{2},...,w_{J+1})^{\prime}$ are estimated by minimizing the pre-treatment discrepancy averaged over all $T_0$ pre-treatment periods:
\begin{equation}
\hat{w}=\argmin_{\substack{w\in\R^J\\\mathbf{1}'w=1}} \frac{1}{T_0}\sum_{t=1}^{T_0}\int_{\calY_0}\left\|\hat{\theta}_{1t}(y)-\sum_{i=2}^{J+1}w_{i}\hat{\theta}_{it}(y)\right\|^{2}dy.
\end{equation}
By averaging over pre-treatment periods and integrating the quadratic discrepancy over the support of the outcome $y$, we move from a point-wise to a distribution-wide weight estimation. This ensures that the weights $w$ are constant across the entire conditional distribution and across time. This global objective function acts as a natural stabilizer, effectively reducing the dimensionality of the optimisation problem and mitigating the risk of over-fitting that typically arises when weighting high-dimensional objects.

In practice, the integral over $\mathcal{Y}_0$ is replaced by a sum over the grid $\mathcal{Y}_{m}$
\begin{equation}
  \label{eq:QP}
  \hat w = \argmin_{\substack{w\in\R^J\\\mathbf{1}'w=1}}
  \frac{1}{T_0\cdot m}\sum_{t=1}^{T_0}\sum_{l=1}^m
  \Bigl\|\hat\theta_{1t}(y_l)
  -\sum_{i=2}^{J+1}w_i\hat\theta_{it}(y_l)\Bigr\|^2.
\end{equation}
This is a quadratic program in $w$ with linear constraints and a positive semidefinite objective.
Let $\hat G\in\R^{J\times J}$ denote the time-averaged Gram matrix with entries
\[
\hat G_{kl}=\frac{1}{T_0\cdot m}\sum_{t=1}^{T_0}\sum_{l'=1}^m \hat\theta_{kt}(y_{l'})'\hat\theta_{lt}(y_{l'})
\]
and $\hat c\in\R^J$ the vector with entries
\[
\hat c_k=\frac{1}{T_0\cdot m}\sum_{t=1}^{T_0}\sum_{l'=1}^m\hat\theta_{1t}(y_{l'})'\hat\theta_{kt}(y_{l'}).
\]
Solving the Lagrangian of~\eqref{eq:QP} yields the closed-form solution~\eqref{eq:wexplicit}, unchanged in form.

\begin{equation}
  \label{eq:wexplicit}
  \hat w = \hat G^{-1}\hat c
  -\hat G^{-1}\mathbf{1}
  \,\frac{\mathbf{1}'\hat G^{-1}\hat c-1}{\mathbf{1}'\hat G^{-1}\mathbf{1}}.
\end{equation}

\begin{remark}[Ridge regularisation]\label{rem:ridge}
When the condition number $\kappa(\hat G)$ is large, the closed-form solution~\eqref{eq:wexplicit} may be numerically unstable. A practical remedy is to replace $\hat G$ by the ridge-regularized matrix $\hat G_\lambda = \hat G + \lambda I_J$ for some $\lambda>0$, yielding
\begin{equation}
  \label{eq:wridge}
  \hat w_\lambda = \hat G_\lambda^{-1}\hat c
  -\hat G_\lambda^{-1}\mathbf{1}
  \,\frac{\mathbf{1}'\hat G_\lambda^{-1}\hat c-1}{\mathbf{1}'\hat G_\lambda^{-1}\mathbf{1}}.
\end{equation}
The regularisation introduces a bias of order $O(\lambda\|w^*\|)$ in the weights, which shrinks $\hat w_\lambda$ toward $J^{-1}\mathbf{1}$. For $\lambda\to 0$ the solution~\eqref{eq:wridge} converges to~\eqref{eq:wexplicit} whenever $\hat G$ is invertible. In practice, $\lambda$ should be chosen small enough that $\lambda/\sigma_{\min}(\hat G)\ll 1$, so that the numerical benefit outweighs the shrinkage bias. The theoretical results of Section~\ref{sec:asymptotics} continue to hold for $\hat w_\lambda$ whenever $\lambda = o(n^{-1/2})$.
\end{remark}

The counterfactual DR parameter for post-treatment period $t\in\{T_0+1,\ldots,T_0+T_1\}$ is estimated as:
\begin{equation}
\hat{\theta}_{1,t}^{0}(y)=\sum_{i=2}^{J+1}\hat{w}_{i}\hat{\theta}_{i,t}(y), \label{eq:counter_param}
\end{equation}
which exploits the pre-treatment balance enforced by the weight estimator~\eqref{eq:QP} (see Remark~\ref{rem:balance}).
The counterfactual conditional distribution function is then $\hat{F}_{1,t}^{0}(y|x)=\Lambda(x^{\prime}\hat{\theta}_{1,t}^{0}(y))$
and the treatment effect estimator for post-period $t$ is:
\begin{equation*}
\hat{f}_t(x)=\int_{\calY_0}[\Lambda(x^{\prime}\hat{\theta}_{1,t}(y))-\Lambda(x^{\prime}\hat{\theta}_{1,t}^{0}(y))]^{2}dy, \label{eq:f_est}
\end{equation*}
approximated in practice by a Riemann sum over $\mathcal{Y}_{m}$.

\begin{remark}[Direct PTP estimator]\label{rem:direct_ptp}
An alternative estimator uses Assumption~1 directly,
\[
  \tilde\theta^{0}_{1,t}(y)
  \;=\;
  \hat\theta_{1,T_0}(y)
  +\sum_{i=2}^{J+1}\hat w_{i}
  \bigl(\hat\theta_{i,t}(y)-\hat\theta_{i,T_0}(y)\bigr),
\]
and is consistent under PTP without requiring perfect pre-treatment balance.
However, it incurs additional variance relative to $\hatthetazerot{t}$
in~\eqref{eq:counter_param}, because pre-period estimation errors in
$\hat\theta_{1,T_0}$ and $\hat\theta_{i,T_0}$ propagate into the counterfactual.
Under perfect balance the two estimators are asymptotically equivalent, so
$\hatthetazerot{t}$ is preferred on efficiency grounds. Corollary~\ref{cor:orthogonality}
identifies the weaker condition under which $\hat f_t(x)$ remains a consistent estimator of $f_t(x)$ even when pre-treatment balance is imperfect.
\end{remark}

\subsection{Confidence interval for $f_t(x)$}
Since $f_t(x)\ge0$ by construction, we report a one-sided confidence interval. When the treatment
effect is non-zero, $f_t(x)>0$, the asymptotic normality of $\hat f_t(x)$
(established in Theorem~\ref{thm:main}(c) in Section~\ref{sec:asymptotics}) yields,
at significance level $\alpha$, the one-sided $(1-\alpha)$ confidence interval
\begin{equation}
  \label{eq:CI}
  CI(t,x) := \left[\,\max\!\Bigl\{0,\ \hat{f}_t(x)-u_{1-\alpha}\,\frac{\hat\sigma_t(x)}{\sqrt{n}}\Bigr\},\ \infty\right),
\end{equation}
where $u_{1-\alpha}$ is the $(1-\alpha)$-quantile of the $N(0,1)$-distribution and
\begin{eqnarray*}
  \label{eq:sigmahat}
  \hat{\sigma}^{2}_t(x) &=& \frac{4}{m^2} \sum_{l=1}^{m} \sum_{l'=1}^{m} \hat{\delta}_{x,t}(y_{l}) \hat{\delta}_{x,t}(y_{l'}) \cdot \hat{K}_t(y_l, y_{l'}) \\
  \hat\delta_{x,t}(y_l)&=&\Lambda(x'\hat\theta_{1,t}(y_l))-\Lambda(x'\hatthetazerot{t}(y_l))
\end{eqnarray*}
and $\hat K_t$ is the covariance kernel for post-period $t$:
\begin{eqnarray}\label{sandwich}
\hat{K}_t(y_l, y_{l'}) &=& \lambda(x'\hat{\theta}_{1,t}(y_l)) \lambda(x'\hat{\theta}_{1,t}(y_{l'})) \cdot \nonumber \\
&& \frac{n}{n_{1,t}^2} \sum_{k=1}^{n_{1,t}} \psi_{1,t,k}(y_l) \psi_{1,t,k}(y_{l'}) \cdot \left[ x'\hat{\mathcal{I}}_{1,t}(y_l)^{-1} X_{1,t,k} \right] \cdot \left[ X_{1,t,k}' \hat{\mathcal{I}}_{1,t}(y_{l'})^{-1} x \right] \nonumber \\
&+& \lambda(x'\hatthetazerot{t}(y_l)) \lambda(x'\hatthetazerot{t}(y_{l'})) \cdot \Biggl( \\
&& \sum_{i=2}^{J+1} \hat{w}_i^2 \left( \frac{n}{n_{i,t}^2} \sum_{k=1}^{n_{i,t}} \psi_{i,t,k}(y_l) \psi_{i,t,k}(y_{l'}) \cdot \left[ x'\hat{\mathcal{I}}_{i,t}(y_l)^{-1} X_{i,t,k} \right] \cdot \left[ X_{i,t,k}' \hat{\mathcal{I}}_{i,t}(y_{l'})^{-1} x \right] \right) \nonumber \\
&+& \hat{\tilde\Theta}_{x,t}(y_l)'\,\hat{V}_w\,\hat{\tilde\Theta}_{x,t}(y_{l'}) \Biggr) \nonumber.
\end{eqnarray}
The generalized residuals (scores) are $$\psi_{i,t,k}(y) = \dfrac{\lambda(X_{i,t,k}'\hat\theta_{i,t}(y))}{\Lambda(X_{i,t,k}'\hat\theta_{i,t}(y))\,(1-\Lambda(X_{i,t,k}'\hat\theta_{i,t}(y)))}\bigl(\mathbf{1}\{Y_{i,t,k} \le y\} - \Lambda(X_{i,t,k}'\hat{\theta}_{i,t}(y))\bigr),$$ the standardized score residual of the binary log-likelihood~\eqref{eq:DR_estim}.
The vector $\hat{\tilde\Theta}_{x,t}(y_l)=\bigl(x'\hat\theta_{2,t}(y_l),\ldots,x'\hat\theta_{J+1,t}(y_l)\bigr)'\in\mathbb{R}^J$
collects the scalar projections of the period-$t$ donor parameters onto~$x$,
and $\hat{V}_w$ is an estimator for $V_w$
from Theorem~\ref{thm:weights}.
Moreover, $$\widehat{\mathcal{I}}_{i,t}(y)=\frac{1}{n_{i,t}}\sum_k \dfrac{\lambda(X_{i,t,k}'\hat\theta_{i,t}(y))^2}{\Lambda(X_{i,t,k}'\hat\theta_{i,t}(y))\,(1-\Lambda(X_{i,t,k}'\hat\theta_{i,t}(y)))}
X_{i,t,k}X_{i,t,k}'$$ is the empirical Fisher information at $y$ (the sample analogue of $\mathcal{I}_{i,t}$ in Lemma~\ref{lem:DRconv}).
The first term in~\eqref{sandwich} captures DR estimation uncertainty
for the treated group in period $t$; the second captures uncertainty from the control groups;
the third captures weight estimation uncertainty, corresponding to the
$\Thetamat_t(y)'V_w\Thetamat_t(y')$ component of~\eqref{eq:cov0}. We estimate $V_w$ by the following plug-in estimator. With $\hat P=\partial w^*/\partial c$ at $\hat G$ (see~\eqref{eq:jacobian_w}) and $\hat\Theta^{(s)}(y)=(\hat\theta_{2,s}(y),\ldots,\hat\theta_{J+1,s}(y))'\in\R^{J\times p}$, define the cell loadings $\hat B_{i,s}(y)\in\R^{J\times p}$ (under Assumption~\ref{ass:span})
\begin{equation}\label{eq:Bload}
\hat B_{1,s}(y)=\hat P\,\hat\Theta^{(s)}(y),\qquad
\hat B_{i,s}(y)=-\hat w_i\,\hat P\,\hat\Theta^{(s)}(y)\quad(i=2,\ldots,J+1),
\end{equation}
and the plug-in estimator is the score-outer-product sum
\begin{equation}\label{eq:Vwplugin}
\hat V_w=\frac{1}{T_0^{2}m^{2}}\sum_{s=1}^{T_0}\sum_{i=1}^{J+1}\frac{n}{n_{i,s}^{2}}\sum_{k=1}^{n_{i,s}}\hat d_{i,s,k}\,\hat d_{i,s,k}',\qquad
\hat d_{i,s,k}=\sum_{l}\psi_{i,s,k}(y_l)\,\hat B_{i,s}(y_l)\,\widehat{\mathcal{I}}_{i,s}(y_l)^{-1}X_{i,s,k}\in\R^{J},
\end{equation}
with the score residual $\psi$ and Fisher information $\widehat{\mathcal{I}}$ of~\eqref{sandwich}; consistency follows from Theorem~\ref{thm:weights}.

The interval~\eqref{eq:CI} achieves nominal one-sided $(1-\alpha)$ coverage
asymptotically when $f_t(x)>0$. When $f_t(x)=0$, $\hat\sigma^2_t(x)$ degenerates to zero
and the lower bound collapses to $0$; testing $H_0:f_t(x)=0$ requires a separate procedure,
described in Section~\ref{sec:inference}.

\subsection{Inference: Supremum Test for $H_0: f(x,\calY_0)=0$}
\label{sec:inference}
%% ---------------------------------------------------------------

The null hypothesis $H_0:f_t(x,\calY_0)=0$ is equivalent to $H_0:\Delta_t(y|x)=0$ for
(Lebesgue-almost) all $y\in\calY_0$. We test this hypothesis via the supremum of the pointwise CDF differences over $\calY_0$.

The test statistic for post-treatment period $t$ and region $\calY_0$ is
\begin{equation}
  \label{eq:Tn}
  T_n(x,t,\calY_0) := \sqrt{n} \cdot \sup_{y_l\in\calY_0}\,|\hat{\delta}_{x,t}(y_l)|,
\end{equation}
where $\hat\delta_{x,t}(y_l):=\hat F_{1,t}(y_l|x)-\hat F^0_{1,t}(y_l|x)$. The full-distribution test uses $\calY_0=\calY_m$; a focused test restricts to a proper subset such as the minimum-wage corridor. The asymptotic theory applies to both with $m$ replaced by $|\{l:y_l\in\calY_0\}|$.

Critical values are obtained by 
%a bootstrap:
simulation of the Gaussian process $Z_x$ from Remark~\ref{thm:sup}:

\begin{enumerate}
  \item On a previously chosen grid $y_l$ estimate the covariance kernel matrix by the sandwich estimator $\hat{K}(y_l, y_{l'})$ from \eqref{sandwich}.
  \item Compute the Cholesky factorisation $\hat K = LL'$ of the covariance matrix.
  \item Draw $S$, e.g. $S=10{,}000$, independent realisations $B^{(s)} = L\varepsilon^{(s)}$, $\varepsilon^{(s)}\sim\mathcal{N}(0,I_m)$, and set $T^{(s)} = \sup_l|B^{(s)}_l|$.
  \item Reject $H_0$ at level $\alpha$ if $T_n(x)>\hat c_{1-\alpha}$, the empirical $(1-\alpha)$-quantile of $\{T^{(s)}\}_{s=1}^S$.
\end{enumerate}
The $p$-value is $\hat p = S^{-1}\sum_s\mathbf{1}\{T^{(s)}\ge T_n(x)\}$. If Stage~1 rejects $H_0$, the confidence interval~\eqref{eq:CI} provides quantitative precision for $f(x)$.

\subsection{Pre-Trend Test}\label{rem:multipleperiods}
When $T_0\geq 2$, evidence for the validity of PTP can be given by pre-testing. As usual in difference-in-difference approaches (\citealp{Lechner2011}), the idea is to extrapolate from the time before the treatment (where all groups are non-treated) to the time after, where the counterfactual outcome of the treated group is not observed.

The tests can be formed in two ways: for a specific $x$ with the
statistic~\eqref{eq:Tn}, or for all $x$ simultaneously via the modification below.
In the focused form, each pre-period $t\in\{1,\ldots,T_0-1\}$ is treated as a
pseudo-post period and $T_n(x,t,\calY)$ from~\eqref{eq:Tn} is computed with weights
$\hat w$ from periods $1,\ldots,t-1$ and the counterfactual $\hat\theta^0_{1,t}$
compared with $\hat\theta_{1,t}$; under Assumption~\ref{ass:span} it should not exceed
$\hat c_{1-\alpha}$, and exceedance signals a pre-treatment violation. We apply this
with $T_0=3$ in Section~\ref{sec:empirics_nj}.

The second form tests PTP for all $x$ at once. If the test does not reject (as is the case in our empirical application), there is direct evidence for a parallel trend structure over the whole parameter space. For strictly increasing $\Lambda$,
$\delta_{x,t}(y)=0$ for every $x$ iff $\theta_{1,t}(y)=\theta^0_{1,t}(y)$, so one
replaces $|\hat\delta_{x,t}(y_l)|$ in~\eqref{eq:Tn} by the parameter-difference norm:
\begin{equation}\label{eq:Tnall}
  \tilde T_n(t,\calY_0):=\sqrt{n}\,\sup_{y_l\in\calY_0}\bigl\|\hat\theta_{1,t}(y_l)-\hat\theta^0_{1,t}(y_l)\bigr\|.
\end{equation}
Critical values follow the
same simulation scheme as for $T_n$, drawing the vector-valued process and taking the
supremum of its norm.

A rejection of the
full-distribution test $T_n(x,t,\calY)$ that does not extend to
$T_n(x,t,\calY_0)$ provides evidence that any PTP violation is driven
by distributional shifts outside $\calY_0$, and is therefore consistent
with Assumption~\ref{ass:span} on $\calY_0$ as required for
$\hat f_t(x,\calY_0)$ to be a valid causal estimand.

\section{Asymptotic Theory}
\label{sec:asymptotics}
%% ---------------------------------------------------------------
This section establishes the asymptotic foundations for the estimators and the tests introduced in Sections~\ref{sec:framework}. All proofs are deferred to the Appendix. We maintain the following conditions throughout:

\begin{assumption}[Sampling]
\label{ass:sampling}
For each $(i,t)$, $\{(Y_{itk},X_{itk})\}_{k=1}^{n_{it}}$ are i.i.d.; different
cells are independent; $n_{it}/n\to r_{it}\in(0,\infty)$.
\end{assumption}

\begin{assumption}[DR Regularity]
\label{ass:regularity}
For each $(i,t)$ and $y\in\calY_0$, $\theta_{it}(y)$ uniquely maximizes the expected
log-likelihood and the score satisfies conditions for the functional CLT in
$\linf(\calY_0)^p$, see \citet{ChernozhFVMelly2013}, Condition DR.
\end{assumption}

\begin{assumption}[Gram Matrix]
\label{ass:gram}
The population Gram matrix $G^*\in\R^{J\times J}$, defined as the population analogue of $\hat G$ in~\eqref{eq:QP},
\[
  G^*_{kl} = \frac{1}{T_0}\sum_{t=1}^{T_0}\int_{\calY_0}\theta_{kt}(y)'\theta_{lt}(y)\,dy,
\]
is positive definite.
\end{assumption}

\begin{assumption}[Pre-Treatment Balance]
\label{ass:span}
The pre-treatment DR parameter function of the treated group lies in the closed linear span of the control groups' pre-treatment DR parameter functions in $L^{\infty}(\calY_0)^{p}$, i.e., there exist weights $w^{*}=(w^{*}_{2},\ldots,w^{*}_{J+1})'$ with $\mathbf{1}'w^{*}=1$ such that
\[
  \theta_{1,t}(y) = \sum_{i=2}^{J+1}w^{*}_{i}\,\theta_{i,t}(y) \qquad \text{for almost all } y\in\calY_0 \text{ and all } t=1,\ldots,T_0.
\]
\end{assumption}

\begin{remark}
Assumption~\ref{ass:span} is the functional analogue of the classical SC span condition: the treated unit's pre-treatment path lies in the convex (here: affine) hull of the donor pool in the DR parameter space. It implies perfect pre-treatment balance in all pre-periods, so that the weight estimator~\eqref{eq:QP} is consistent for $w^*$ and $\hatthetazerot{t}$ is a consistent estimator of $\thetazerot{t}$ for each post-period $t$. Together with Assumption~\ref{ass:gram}, it also implies that $w^*$ is the unique minimum-norm solution to the pre-treatment balance equations. This assumption can be relaxed, in particular under an orthogonality condition (Corollary \ref{cor:orthogonality}) and using a different estimator (Remark \ref{rem:direct_ptp}). Moreover, it can be tested, see Section~\ref{rem:multipleperiods}.
\end{remark}

\begin{remark}
Assumption~\ref{ass:gram} ensures that the control groups are sufficiently distinct to identify the weights uniquely.
\end{remark}

The technical foundation rests on three results (Lemmas~\ref{lem:DRconv}, \ref{lemma:gram} and Theorem~\ref{thm:weights}), stated and proved in Appendix~\ref{app:lemmas}: weak convergence of the DR estimators in $l^\infty(\calY_0)^p$ (Lemma~\ref{lem:DRconv}), asymptotic normality of the Gram-matrix estimator (Lemma~\ref{lemma:gram}), and asymptotic normality of the weight estimator (Theorem~\ref{thm:weights}), from which consistency of the plug-in $\hat V_w$ follows by continuity. We now state the main result. Fix a post-treatment period $t\in\{T_0+1,\ldots,T_0+T_1\}$. Let $\Theta_t(y):=(\theta_{2,t}(y),\dots,\theta_{J+1,t}(y))^{\prime}\in\mathbb{R}^{J \times p}$ collect the period-$t$ DR parameters of all control groups.

\begin{theorem}[Asymptotic distribution of counterfactual and treatment effect]
\label{thm:main}
Under Assumptions~\ref{ass:ptp}--\ref{ass:span}, for a fixed evaluation point $x\in\calX$ and post-treatment period $t$:

\textup{(a)} $\sqrt{n}(\hatthetazerot{t}(\cdot)-\thetazerot{t}(\cdot))\wto\G^0_t(\cdot)$
in $\linf(\calY_0)^p$, where
\begin{equation}
  \label{eq:cov0}
  \Cov(\G^0_t(y),\G^0_t(y'))
  =\underbrace{\sum_{i=2}^{J+1}\frac{(w^*_i)^2}{r_{i,t}}
    \Cov(\G_{i,t}(y),\G_{i,t}(y'))}_{\text{DR estimation error (post-period $t$)}}
  +\underbrace{\Thetamat_t(y)'V_w\,\Thetamat_t(y').}_{\substack{\text{weight estimation error}\\[2pt] \scriptstyle V_w\text{ aggregates the pre-period }r_{i,s}^{-1}\Phi_{i,s}\text{ of~\eqref{eq:SigmaGc}}}}
\end{equation}

\textup{(b)} Define the common-scale process $\mathbb{H}_{it}:=r_{it}^{-1/2}\mathbb{G}_{it}$. For $Z_{n,x,t}(y) := \sqrt{n}(\hat\delta_{x,t}(y)-\delta_{x,t}(y))$, it holds $Z_{n,x,t}(\cdot) \wto Z_{x,t}(\cdot) = \lambda(x'\theta_{1,t}(\cdot)) \cdot x' \mathbb{H}_{1,t}(\cdot) - \lambda(x'\thetazerot{t}(\cdot)) \cdot x' \mathbb{G}^0_t(\cdot)$.

\textup{(c)} If $f_t(x)>0$:
$\sqrt{n}(\hat{f}_t(x)-f_t(x))\dto\mathcal{N}(0,\sigma^2_t(x))$,
where $\sigma^2_t(x)$ is the variance of
$\dot\phi_{(\theta_{1,t}(y),\thetazerot{t}(y))}(\H_{1,t},\G^0_t)$ and
\[
  \dot\phi(h,h^0)=2\int_{\calY_0}
  [\Lambda(x'\theta_{1,t}(y))-\Lambda(x'\thetazerot{t}(y))]
  [\lambda(x'\theta_{1,t}(y))x'h-\lambda(x'\thetazerot{t}(y))x'h^0]\,dy.
\]
\end{theorem}

\begin{remark}
In Theorem~\ref{thm:main}(c), the condition $f_t(x)>0$ is essential: under $H_0:\delta(y)=0$, the derivative
$\dot\phi$ vanishes and $\sigma^2(x)=0$, so this result does not provide a
valid test for $H_0$. It underpins the confidence interval~\eqref{eq:CI} in Section~\ref{sec:framework}, which is valid conditional on $f_t(x)>0$ having been established by the supremum test of Section~\ref{sec:inference}. In fact, Theorem~\ref{thm:main}(b) is the theoretical basis for this test. Moreover, the covariance kernel $K$ of $Z(\cdot)$ is estimated with the sandwich formula for the confidence interval~\eqref{eq:CI}.
\end{remark}

\begin{remark}
The two-component covariance in~\eqref{eq:cov0} corrects a common implicit error:
weight estimation uncertainty and DR estimation uncertainty are both $O_p(n^{-1/2})$
and are both included in the variance of the counterfactual. The two terms are
uncorrelated because $\hat w$ depends only on pre-period data while the first term involves post-period data (Assumption~\ref{ass:sampling}). The terms can be consistently estimated (as in \eqref{sandwich}) by replacing all population quantities with their sample counterparts and estimating the covariance kernels of $\mathbb{G}_{it}$ via the sandwich formula for binary outcome regressions.
\end{remark}

\begin{remark}[Supremum test]\label{thm:sup}
Theorem~\ref{thm:main}(b) immediately yields the asymptotic null distribution of~\eqref{eq:Tn}: under $H_0:\Delta_t(\cdot|x)=0$ on $\calY_0$, the continuous mapping theorem gives
\[
  T_n(x,t,\calY_0) \dto \sup_{y\in\calY_0}|Z_{n,x,t}(y)|,
\]
where $Z_{n,x,t}$ is the Gaussian process from Theorem~\ref{thm:main}(b). Under $H_1:f_t(x,\calY_0)>0$, $T_n(x,t,\calY_0)\to\infty$ in probability.
\end{remark}

\begin{remark}[Limiting distribution of the all-$x$ pre-trend test]\label{rem:allx_kernel}
Theorem \ref{thm:main}.(a) yields with the continuous mapping theorem
\begin{equation}\label{eq:Tnall_limit}
  \tilde T_n(t,\calY_0)\;\dto\;\sup_{y\in\calY_0}\left|\mathbb{H}_{1,t}(y)-\G^0_t(y)\right|,
\end{equation}
The covariance kernel of this limit process is obtained from~\eqref{sandwich} by dropping the scalar
link factors $\lambda(x'\cdot)$ and replacing each projection $x'\hat{\mathcal I}^{-1}X_k$
by $\hat{\mathcal I}^{-1}X_k$, so that the kernel is $\R^{p\times p}$-valued and the
weight term becomes the $p\times p$ matrix
$\hat\Theta^{(t)}(y_l)'\hat V_w\,\hat\Theta^{(t)}(y_{l'})$.
\end{remark}

\begin{remark}[Confidence interval]
Theorem \ref{thm:main}.(c) yields with some algebraic manipulations
\begin{equation*}
P(f_t(x) \in CI(t,x))\dto 1-\alpha
\end{equation*}
for $f_t(x) > 0.$. Because the bound is reported only after the supremum test rejects, its coverage is to be understood conditional on that rejection.
\end{remark}

\subsection{Imperfect Pre-Treatment Balance}
\label{subsec:imperfect_balance}
If Assumption~\ref{ass:span} fails with residual
$r_n(y):=\theta_{1,T_0}(y)-\sum_{i=2}^{J+1}w_i^*\theta_{i,T_0}(y)\neq 0$,
the counterfactual estimator $\hat{f}_t(x)$ is generally not a consistent estimator of $f_t(x)$. Instead, it converges to a pseudo-value $\tilde f_t(x)\neq f_t(x)$.
The following corollary gives a precise condition under which consistency for $f_t(x)$ is nevertheless restored.

\begin{corollary}[Robustness to imperfect balance]\label{cor:orthogonality}
Suppose Assumption~\ref{ass:ptp} and Assumptions~\ref{ass:sampling}--\ref{ass:gram} hold,
$f_t(x)>0$, and Assumption~\ref{ass:span} fails with residual sequence
$r_n(y):=\theta_{1,T_0}(y)-\sum_{i=2}^{J+1}w_i^*\theta_{i,T_0}(y)$.
Under Assumption~\ref{ass:ptp}, the true counterfactual satisfies
$\thetazerot{t}(y)=\sum_i w_i^*\theta_{i,t}(y)+r_n(y)$, so the scaled bias of
$\hat{f}_t(x)$ relative to $f_t(x)$ is
\begin{eqnarray}\label{eq:bias_decomp}
  \sqrt{n}\bigl[\tilde{f}_t(x)-f_t(x)\bigr] \nonumber
  &=& \sqrt{n}\Bigl[2\langle r_n,\Delta_t\rangle_{x}
    + \int_{\calY_0}\left[\lambda(\eta_t(y))^2+\Delta_t(y|x) \lambda'(\eta_t(y))\right](x'r_n(y))^2\,dy
    \Bigr] \\ &+& O(\|r_n\|^3),
\end{eqnarray}
where $\eta_t(y) := x'\thetazerot{t}(y)$ and $\langle r_n,\Delta_t\rangle_{x}:=\int_{\calY_0}\lambda(x'\eta_t(y))\,
x'r_n(y)\,\Delta_t(y\mid x)\,dy$.
Then $\sqrt{n}(\hat{f}_t(x)-f_t(x))\dto\mathcal{N}(0,\sigma^2_t(x))$ as in
Theorem~\ref{thm:main}(c) if and only if \eqref{eq:bias_decomp} converges to zero.
Two sufficient conditions are:
\begin{itemize}
  \item[(a)] $\langle r_n,\Delta_t\rangle_x=o(n^{-1/2})$ and $\|r_n\|^2=o(n^{-1/2})$;
  \item[(b)] $\langle r_n,\Delta_t\rangle_x=0$ and $\|r_n\|=o(n^{-1/4})$.
\end{itemize}
Condition (b) is strictly weaker than requiring $\|r_n\|=o(n^{-1/2})$, which
would be needed without orthogonality to control the first-order term alone.
\end{corollary}

\begin{remark}[Economic interpretation]
\label{rem:orthogonality_econ}
The orthogonality condition $\langle r_n,\Delta\rangle_{x}=0$ holds whenever
the pre-treatment misfit $r_n(y)$ is concentrated in regions of the outcome
space where the treatment effect $\Delta(y\mid x)$ is negligible, and vice
versa.  In the New Jersey application, $\Delta(y\mid x_{10})$ is sharply
concentrated in the minimum-wage corridor $[\log 4.25,\log 5.10]$.
Corollary~\ref{cor:orthogonality} therefore implies that the estimate for
young low-education workers is robust to pre-treatment misfit outside this
corridor, regardless of its magnitude.
\end{remark}

\section{Simulation Study}
\label{sec:simulations}
%% ---------------------------------------------------------------

To mimic the empirical application with a Mincer's earnings equation, we consider $J+1=5$ groups and two time periods. For $i=1,\ldots,5$ and $j=0,1$ we have the model:
\begin{equation*}
Y_{it} = \beta_{0,it} + \beta_{1,it} X_{1,it} + \beta_{2,it} X_{2,it} + \beta_{3,it} X_{3,it} + \epsilon_{it}
\end{equation*}
with i.i.d.\ $X_{1,it},X_{2,it},X_{3,it} \sim N(0,1)$ and $\epsilon_{it}$ a noise distribution. This leads to $F_{it}(y|x) = \Phi\bigl(\theta_{it}(y)'x\bigr)$ with $\theta_{it}(y) = \bigl(y-\beta_{0,it},-\beta_{1,it},-\beta_{2,it},-\beta_{3,it}\bigr)$.

The $J=4$ control groups have parameter vectors $\boldsymbol{\beta}_i = \bar{\boldsymbol{\beta}} + c\,\mathbf{e}_i$ for $i=2,\ldots,5$, where $\bar{\boldsymbol{\beta}}=(1,1,1,1)'$, $c=0.8$, and $\mathbf{e}_i$ is the $i$-th standard basis vector in $\mathbb{R}^4$. Each control group thus differs from the others in exactly one coefficient. Under the null hypothesis, the treated group has $\boldsymbol{\beta}_1 = \frac{1}{J}\sum_{i=2}^{5}\boldsymbol{\beta}_i = (1.2,1.2,1.2,1.2)'$, so that PTP holds exactly with equal true weights $w^*_i=1/4$. This design ensures that the Gram matrix $G^*$ has full rank $J=4$ and condition number $\kappa(G^*)=55$, satisfying Assumption~\ref{ass:gram}.

Under the alternative hypothesis ($H_1$), the treated group's post-treatment
outcome receives a \emph{covariate-heterogeneous} effect $\Delta\cdot \beta_{1,11}$, i.e.\ a
shift only along the $X_1$ direction, with $\Delta\in\{0,0.1,\ldots,0.5\}$. This shifts
the \emph{conditional} distribution at covariate values with $X_1\neq0$, while leaving
the \emph{unconditional} (marginal) distribution almost unchanged: since $E[X_1]=0$, the
marginal mean and median are preserved and only the spread changes slightly. The design
thus isolates a setting in which conditioning on covariates is essential.

The simulation focuses on the size and power properties of the supremum test \eqref{eq:Tn}. We compare three procedures: the conditional supremum test~\eqref{eq:Tn} with correctly specified errors ($\epsilon_{i,t}\sim N(0,1)$, ``Probit'') and with misspecified errors ($\epsilon_{i,t}\sim Lo(0,1)$, ``Logit''), and an \emph{unconditional} distributional synthetic-control test that applies the same idea to the marginal outcome distribution (intercept-only DR, $p=1$), the natural in-framework analogue of an unconditional approach in the spirit of \citet{Gunsilius2023}. The DR working model uses the probit link throughout. We use a grid of $m=10$ equidistant quantile levels of the outcome between $0.05$ and $0.95$, sample sizes $n_{it} \in \{200, 500, 1000\}$, and $N_{\rm mc}=1000$ Monte Carlo replications. Critical values are obtained by Gaussian process simulation with $S=10{,}000$ draws using the estimated covariance matrix $\hat K$.

The conditional tests are evaluated at $x=(1,1,0,0)'$ (so that $x_1=1$ and the
conditional shift equals $\Delta$); the unconditional test targets the marginal CDF.
Figure~\ref{fig:power_curve} shows that all three procedures control size at $\Delta=0$, whereas the unconditional test is conservative.
Under $H_1$, the conditional tests detect the effect with power increasing in $n$ and
$\Delta$ (slightly higher under the correctly specified probit link than under the
misspecified logit link), whereas the unconditional test has low power: The
covariate-heterogeneous effect is invisible in the marginal distribution. Conditioning
on covariates is therefore necessary to detect effects that are heterogeneous across the
covariate space.

\begin{figure}[ht]
\centering
\includegraphics[width=\textwidth]{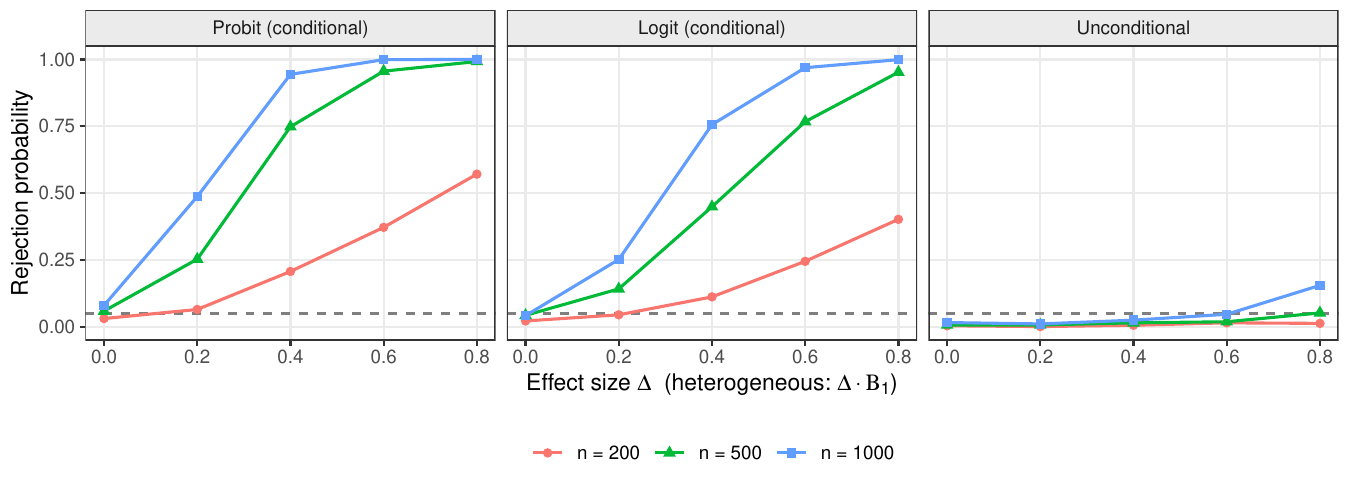}
\caption{Simulated size and power under the covariate-heterogeneous effect
$\Delta\cdot \beta_{1,11}$, by sample size $n$. Left: conditional test, correctly specified
(probit, $N(0,1)$ errors). Middle: conditional test, misspecified link (logistic
$Lo(0,1)$ errors). Right: unconditional test on the marginal distribution. Dashed line:
nominal level $\alpha=0.05$.}
\label{fig:power_curve}
\end{figure}

Figure~\ref{fig:coverage} reports properties of the one-sided 95\% confidence
interval~\eqref{eq:CI} for $f_t(x)$ from the conditional probit estimation, under $H_1$
($\Delta>0$). Coverage is close to the nominal level across $\Delta$; the small deviations
at the smallest effect sizes reflect the positive finite-sample bias of $\hat f$, which is
a mean of squared quantities. The mean half-width falls with $n$ and grows with $\Delta$,
as expected from the asymptotic theory. In practice the CI is reported only when the
supremum test rejects $H_0$, which already selects cases where $f_0$ is not negligible
relative to estimation error.

\begin{figure}[ht]
\centering
\includegraphics[width=\textwidth]{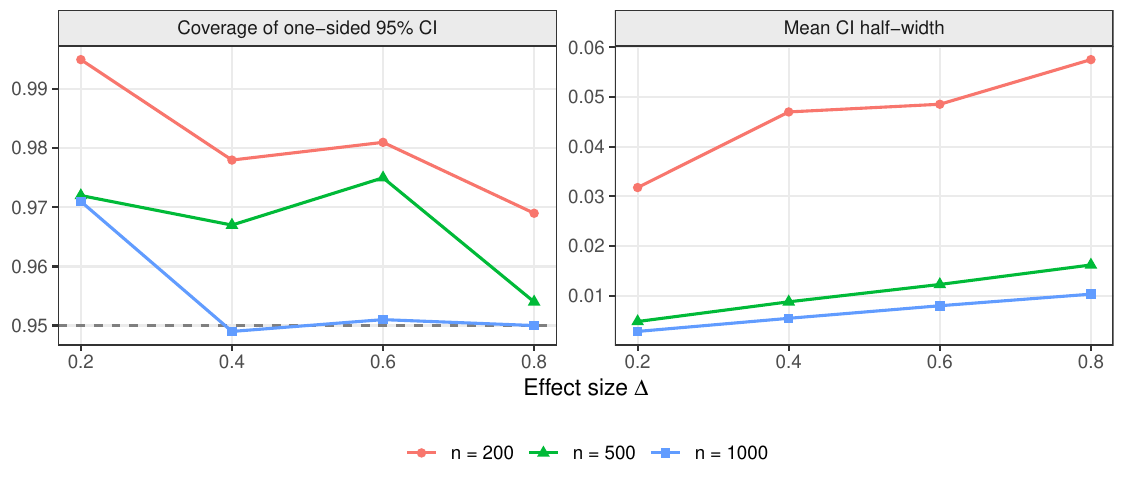}
\caption{Properties of the one-sided 95\% CI~\eqref{eq:CI} for $f_t(x_0)$ under $H_1$
(conditional probit estimation, $N_{\rm mc}=1{,}000$). Left: empirical coverage (dashed:
nominal $0.95$). Right: mean CI half-width. DGP as in Figure~\ref{fig:power_curve}.}
\label{fig:coverage}
\end{figure}

Under $T_0=2$ pre-treatment periods, the $T_0^{-1}$-scaling of $\Sigma_{G,c}$ in Lemma~\ref{lemma:gram} halves the asymptotic variance of the Gram-matrix estimator, and hence of $V_w$, of the weight estimator relative to $T_0=1$. We verified in an auxiliary simulation (same DGP, $T_0=2$, second pre-period drawn i.i.d.\ with the same parameters) that size remains well controlled at the nominal 5\% level and power is modestly improved, consistent with the theoretical efficiency gain.

\section{Empirical Study: New Jersey Minimum Wage 1992}
\label{sec:empirics_nj}
%% ---------------------------------------------------------------

\subsection{Background and Data}

In April 1992, New Jersey raised its state minimum wage from \$4.25 to \$5.05 per
hour (an increase of 19\%) while the federal minimum remained unchanged at \$4.25.
This policy change has been studied in the literature before.
\citet{CardKrueger1994} find no negative employment effects using a DiD comparison
of fast-food restaurants with neighbouring Pennsylvania, contrasting findings in \citep{NeumarkWascher1992}. The DR-SC approach taken here offers a
complementary perspective: rather than average employment effects, we trace the
full conditional wage distribution across different worker types, allowing us to
assess where in the distribution the policy had bite.

The distributional consequences of minimum wage increases have been widely studied:
\citet{DiNardoFortinLemieux1996} document a link between minimum wage erosion and
rising wage inequality using CPS data, and \citet{AutorManningSmith2016} provide
evidence that the minimum wage compresses the lower tail of the wage distribution.
\citet{ChernozhFVMelly2013} study counterfactual wage distributions using distribution
regression. The contribution of the present approach is a clearer causal identification in terms of the workers' characteristics.

The validity of PTP in this application rests on the observation that the three pre-treatment years (1989--1992) span a well-defined common macroeconomic episode: the onset and trough of the 1990--91 recession, followed by an early recovery. This recession depressed low-education wages across all US states, but with differential intensity depending on each state's exposure to interest-rate-sensitive sectors (construction, durable manufacturing) and to the contemporaneous defense-spending contraction \citep{BlanchardKatz1992}. These cross-state differences in cyclical exposure correspond to the state-specific factor loadings $\Lambda_i(y)$ in the latent-factor model of Remark~2.2: New Jersey's exposure profile (concentrated in finance, insurance, and real estate) is approximated by the weighted combination of donor states with Florida (finance and real estate) receiving the largest positive weight and South Dakota and Missouri (manufacturing-heavy, lower cyclical correlation with New Jersey) receiving negative weights, as reported in Table~\ref{tab:weights}. The convincing pre-trend tests in Section~\ref{subsec:pretrends} confirm that this synthetic New Jersey tracks the observed conditional wage distribution across all three pre-treatment years.

We use data from the CPS Basic Monthly files (\citealp{flood2025ipums}) organised into April--March policy years aligned with the reform date: $t=1$ (Apr 1989--Mar 1990), $t=2$ (Apr 1990--Mar 1991) and $t=3$ (Apr 1991--Mar 1992) are pre-treatment ($T_0=3$), and $t=4$ (Apr 1992--Mar 1993) is post-treatment ($T_1=1$). This periodisation aligns the start of the post-period exactly with the April 1992 increase, so that no partially-treated calendar year has to be discarded. We restrict to workers with a valid hourly wage, at most \$99/hour and non-missing education and experience.
The outcome $Y$ is the log nominal hourly wage. The covariate vector $X$ includes a
constant, years of education, years of potential experience, and its square, all
standardised to zero mean and unit variance using the pooled sample. The
standardised experience-squared variable is constructed as the square of standardised
experience, $\widetilde{x_3}=\bigl((\mathrm{exper}-\bar x_2)/s_2\bigr)^2$.
New Jersey policy-year sample sizes are 3{,}852, 3{,}998 and 4{,}031 (pre-treatment) and 3{,}905 (post-treatment).

The donor pool consists of $J=42$ US states, excluding eight states or districts
with their own minimum wage changes during 1989--1993 (Connecticut, Massachusetts,
Rhode Island, New Hampshire, Vermont, Alaska, Hawaii, and the District of Columbia).

As evaluation points, we consider the median worker (education = 12, experience = 10) as well as the four combinations of the $10\%$ and $90 \%$ values of education (10, 16) and experience (2, 37).

\subsection{Weight Estimation and Pre-Treatment Fit}

We estimate the DR parameters on a grid of $m=32$ quantiles of the pooled
log-wage distribution (10th--90th percentiles). The total sample size is $n=381{,}953$ across all group-time cells.
The Gram matrix $\hat G$ is the time-average over the three pre-treatment periods
(see Section~2.2). It has full rank $J=42$, condition number
$\kappa(\hat G)=9.6\times10^4$, and is numerically
stable without regularisation. Table~\ref{tab:weights} reports the five largest
weights by magnitude.

\begin{table}[ht]
\centering
\caption{Five largest synthetic control weights}
\label{tab:weights}
\small
\begin{tabular}{lccccc}
\toprule
Florida & New York & South Dakota & Nevada & Missouri \\
$0.515$ & $0.479$ & $-0.255$ & $0.199$ & $-0.196$ \\
\bottomrule
\end{tabular}
\end{table}

\noindent
Florida receives the largest positive weight, with New York second. The full weight distribution is shown in Figure~\ref{fig:nj_weights}. 20 of 42
states receive negative weights, reflecting that New Jersey's pre-treatment wage
structure is not fully within the convex hull of the donor pool, which the
unconstrained formulation handles naturally.

\begin{figure}[ht]
\centering
\includegraphics[width=.5\textwidth]{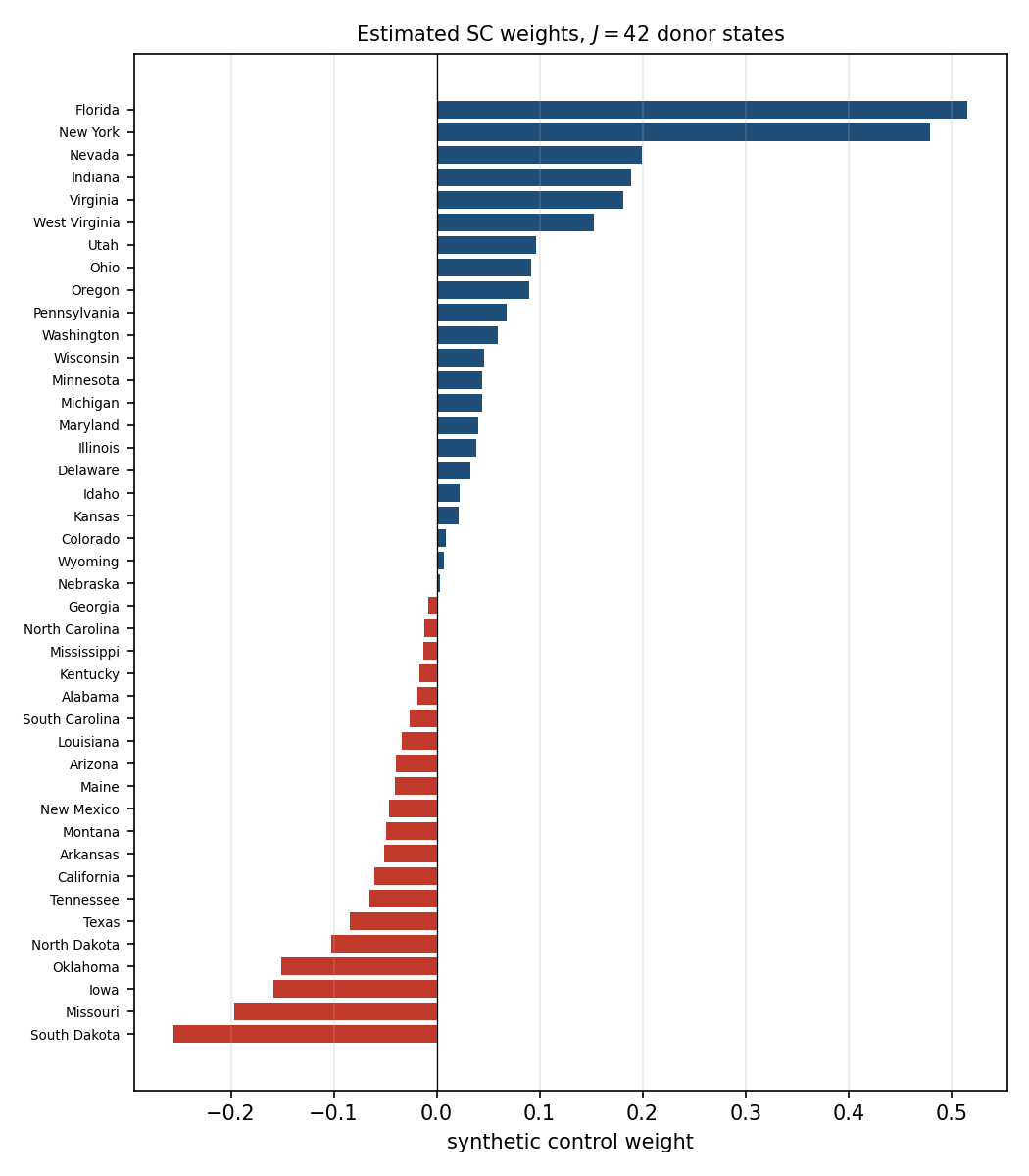}
\caption{Estimated synthetic control weights for all $J=42$ donor states.
Blue bars: positive weights. Red bars: negative weights.}
\label{fig:nj_weights}
\end{figure}

Due to the post-period evidence from \ref{subsec:interpretation}, the evaluation point $x_{10}$ (educ=10, exper=2), reflecting low-education young workers, is of particular interest. Figure~\ref{fig:nj_prefit} shows the pre-treatment fit for this point in all three pre-treatment years.
The synthetic control tracks the observed New Jersey CDF closely across all
pre-treatment years, with maximum sup-distance $\sup_y|\hat F_{\rm NJ,t}(y|x_{10})-
\hat F^0_{\rm NJ,t}(y|x_{10})|\leq 0.037$ in each pre-period.

\begin{figure}[ht]
\centering
\includegraphics[width=.5\textwidth]{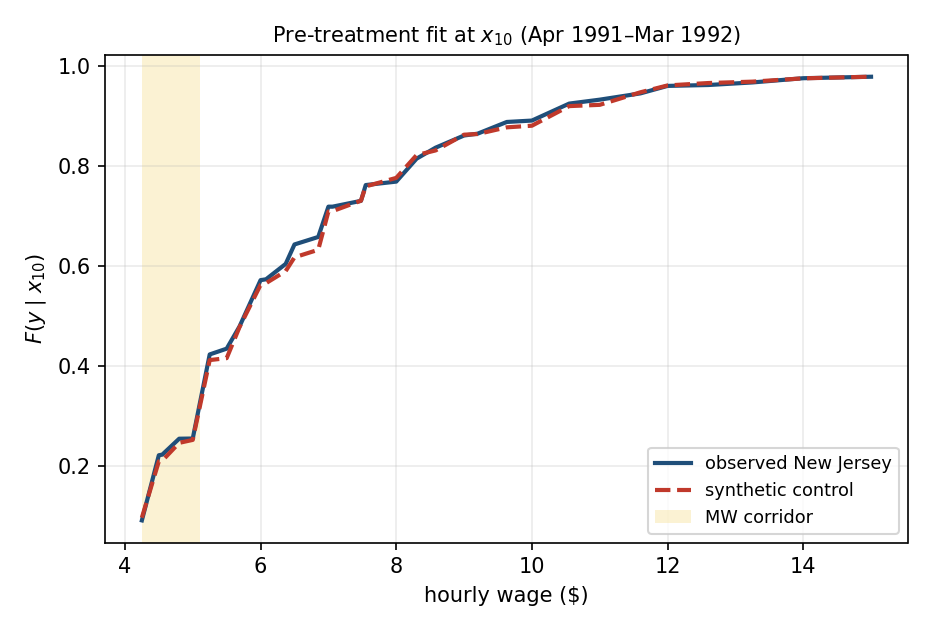}
\caption{Pre-treatment fit at $x_{10}$ (educ=10, exper=2). Observed New Jersey
conditional CDF in the last pre-period (Apr 1991--Mar 1992, solid) and the
synthetic control counterfactual (dashed); the yellow band marks the MW corridor
$[\$4.25,\$5.10]$.}
\label{fig:nj_prefit}
\end{figure}

\subsection{Pre-Trend Tests}\label{subsec:pretrends}

Figure~\ref{fig:nj_pretrend} reports the full-distribution and focused
(MW-corridor) $p$-values of $T_n(x,t,\calY)$ for the two pseudo-post pre-periods
(Apr89/90$\to$Apr90/91 and $\to$Apr91/92) and the post-treatment period
($\to$Apr92/93), at all five evaluation points. The corridor for the focused test is fixed \emph{a priori} to $\calY_0=[\log\$4.25,\log\$5.10]$ by the policy (raising the floor from \$4.25 to \$5.05).

Both pseudo-post tests do not reject $H_0$ at the five evaluation points, on the full distribution and on the
corridor (all $p>0.35$; at $x_{10}$, $p=0.74$ and $p=0.54$ on $\calY$, and
$p=0.81$ and $p=0.70$ on $\calY_0$). The 1990--91 recession is absorbed \emph{within} the twelve-month
windows rather than surfacing at a period transition, and the late-1991
announcement window (contained in the third pre-period) produces no detectable
corridor pre-trend. Unlike a calendar-year periodisation, the design therefore
requires no reconciliation of a full-distribution pre-trend rejection with the
focused test. It seems that the parallel-trends assumption holds directly, both overall and on
the corridor.

Beyond the five evaluation points, we also apply the $x$-simultaneous pre-trend
test~\eqref{eq:Tnall}, which checks parallel trends in the full DR parameter vector and
hence for all covariate values at once. Neither pseudo-post transition rejects at the
$10\%$ level: $\tilde T_n=144.9$ against $\hat c_{0.90}=245.3$ ($p=0.75$) for
Apr89/90$\to$Apr90/91, and $\tilde T_n=209.6$ against $\hat c_{0.90}=241.7$ ($p=0.22$)
for $\to$Apr91/92. So, it is reasonable to assume that parallel trends in parameters holds in the pre-periods not only
at the points of interest but jointly across the covariate space.

\begin{figure}[ht]
\centering
\includegraphics[width=\textwidth]{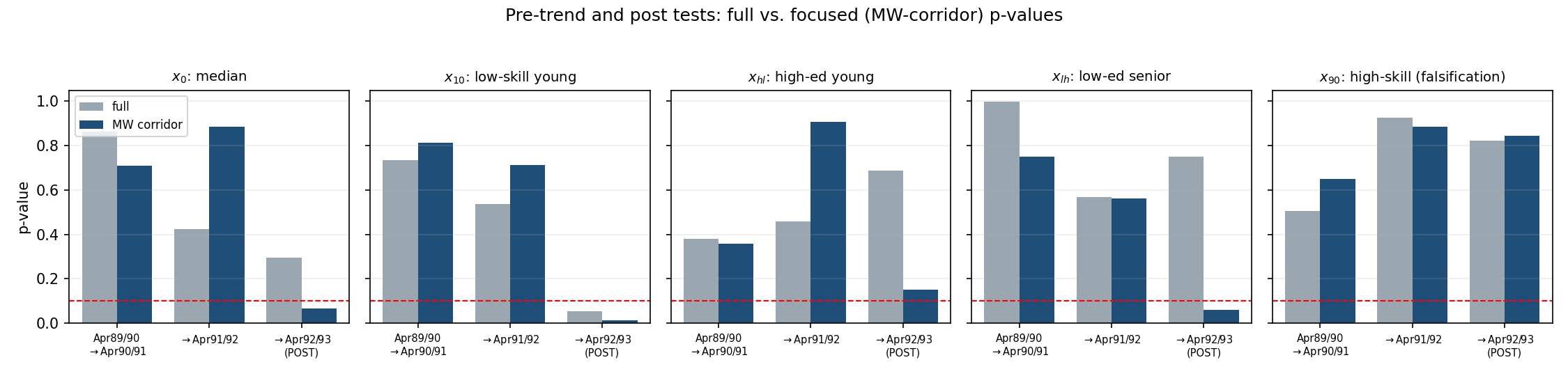}
\caption{Pre-trend (pseudo-post) and post-treatment tests at five covariate
values. Bars show $p$-values of the full-distribution test (grey) and the focused
MW-corridor test (navy) for the two pseudo-post transitions and the post period.
Dashed red line: $\alpha=0.10$. Both pre-trend transitions pass everywhere; the
post effect at $x_{10}$ rejects on the corridor.}
\label{fig:nj_pretrend}
\end{figure}

Figure~\ref{fig:nj_focused} contrasts the full-distribution and focused tests at
$x_{10}$. The focused statistic $T_n(x,t,\calY_0)$ serves two distinct roles.
In the pre-periods it is the empirical counterpart of the orthogonality condition
in Corollary~\ref{cor:orthogonality}: should a full-distribution pre-trend reject
while the corridor test passes, there is evidence that a potential pre-treatment misfit $r_n$ is orthogonal to
the treatment effect $\Delta_t$ on $\calY_0$, and $\hat f_t(x,\calY_0)$ that remains
consistent and asymptotically normal even where the Span Condition fails on the
full support $\calY$. Here both pre-trend tests pass, so this safeguard is not
needed and the corridor effect is identified without it. In the post-period the
focused test instead provides \emph{power}: It concentrates the supremum on the
interval where the policy operates. At $x_{10}$ the full-distribution test is only
marginal ($p=0.053$) while the focused test rejects clearly ($p=0.011$), with the
supremum attained inside the corridor.

\begin{figure}[ht]
\centering
\includegraphics[width=\textwidth]{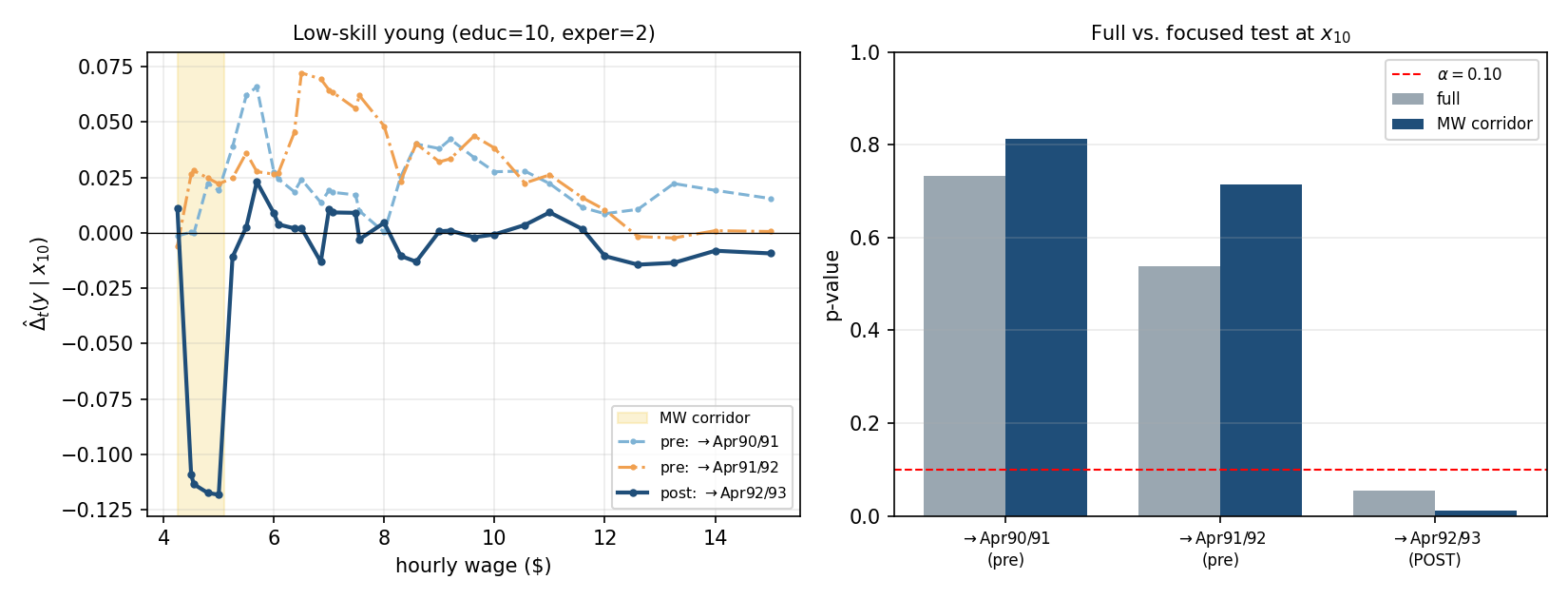}
\caption{Focused supremum test at $x_{10}$ (educ=10, exper=2). \textit{Left:}
$\hat\Delta_t(y\mid x_{10})$ for the two pre-trend transitions and the post period;
yellow band = MW corridor $[\$4.25,\$5.10]$. \textit{Right:} $p$-values of the
full-distribution test (grey) and the focused test (navy). Both pre-trend
transitions pass on $\calY$ and $\calY_0$; in the post period the corridor test
rejects ($p=0.012$) while the full test is only marginal ($p=0.054$). Dashed red
line: $\alpha=0.10$.}
\label{fig:nj_focused}
\end{figure}

\subsection{Main Results: Post-Treatment CDFs and Heterogeneity Patterns for the Treatment Effects}\label{subsec:interpretation}

Figure~\ref{fig:nj_delta} displays the pointwise CDF differences
$\hat\Delta_t(y\mid x):=\hat F_{1,t}(y\mid x)-\hat F^0_{1,t}(y\mid x)$
with 90\% pointwise confidence bands for all five evaluation points.
Table~\ref{tab:results} summarises the scalar treatment effect estimates
$\hat f_t(x,\calY)$ and the supremum test results. The effect is sharply
concentrated in the MW corridor for the young low-educated group $x_{10}$, where
$\hat\Delta_t$ falls by up to $0.12$ inside $[\$4.25,\$5.10]$. Both the focused test ($p_{\calY_0}=0.011$) and the full-distribution test ($p=0.053$) reject at the $10\%$
level. The high-education/high-experience point $x_{90}$ shows no effect on either test
($p=0.83$, $p_{\calY_0}=0.84$). We report the one-sided $90\%$ lower confidence bound
$\mathrm{LCB}=\max\{0,\hat f_t-z_{0.90}\hat{\mathrm{se}}\}$ exactly when the supremum test rejects at $10\%$. This is only the case for $x_{10}$. Figure~\ref{fig:nj_fhat} summarises the treatment effect estimates
visually.

\begin{figure}[ht]
\centering
\includegraphics[width=\textwidth]{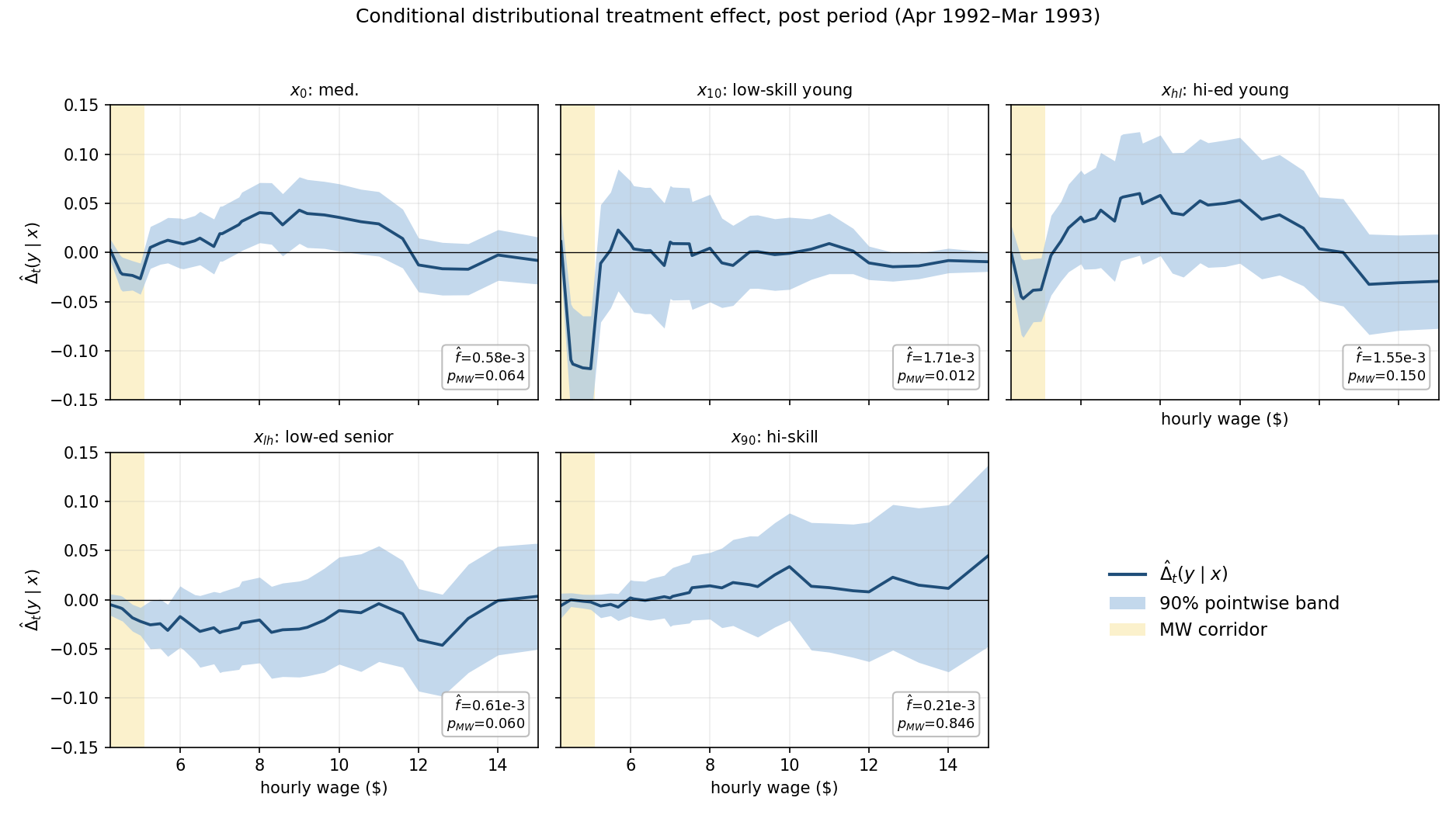}
\caption{Pointwise CDF difference $\hat\Delta_t(y\mid x)$ with 90\% pointwise
confidence band. Yellow band: MW corridor $[\$4.25,\$5.10]$.
Dotted vertical line: new NJ minimum wage $\log(\$5.10)=1.63$.}
\label{fig:nj_delta}
\end{figure}

\begin{table}[ht]
\begin{singlespace}
\centering
\caption{Treatment effect estimates at five covariate values.
$T_n=T_n(x,t,\calY)$ and $\hat p$ are full-distribution;
$\hat p_{\calY_0}$ is the focused MW-corridor test; $\hat c_{0.90}$: $90\%$ GP
critical value ($S=10{,}000$). As $f_t(x)\ge0$, the last column gives the
one-sided $90\%$ confidence interval $[\max\{0,\hat f_t-z_{0.90}\hat{\mathrm{se}}\},\infty)$;
a dash (---) marks a lower bound of $0$, i.e.\ the uninformative interval
$[0,\infty)$ (the test does not reject at $10\%$).}
\label{tab:results}
\small
\setlength{\tabcolsep}{4pt}
\begin{tabular}{@{}lccccccc@{}}
\toprule
$x$ & $\hat f_t\!\times\!10^{-3}$ & $\hat{\mathrm{se}}\!\times\!10^{-3}$ &
$T_n$ & $\hat c_{0.90}$ & $\hat p$ & $\hat p_{\calY_0}$ & 90\% CI $\times10^{-3}$ \\
\midrule
$x_0$ (ed=12, ex=10)   & 0.58 & 0.62 & 26.7 & 34.5 & 0.294 & 0.064 & --- \\
$x_{10}$ (ed=10, ex=2) & 1.71 & 1.13 & 73.0 & 64.8 & 0.054 & 0.012 & $[0.27,\infty)$ \\
$x_{hl}$ (ed=16, ex=2) & 1.55 & 1.99 & 37.1 & 68.7 & 0.686 & 0.150 & --- \\
$x_{lh}$ (ed=10, ex=37)& 0.61 & 1.08 & 28.6 & 59.4 & 0.750 & 0.060 & --- \\
$x_{90}$ (ed=16, ex=37)& 0.21 & 0.76 & 27.6 & 80.2 & 0.824 & 0.846 & --- \\
\bottomrule
\end{tabular}
\end{singlespace}
\end{table}

\begin{figure}[ht]
\centering
\includegraphics[width=.85\textwidth]{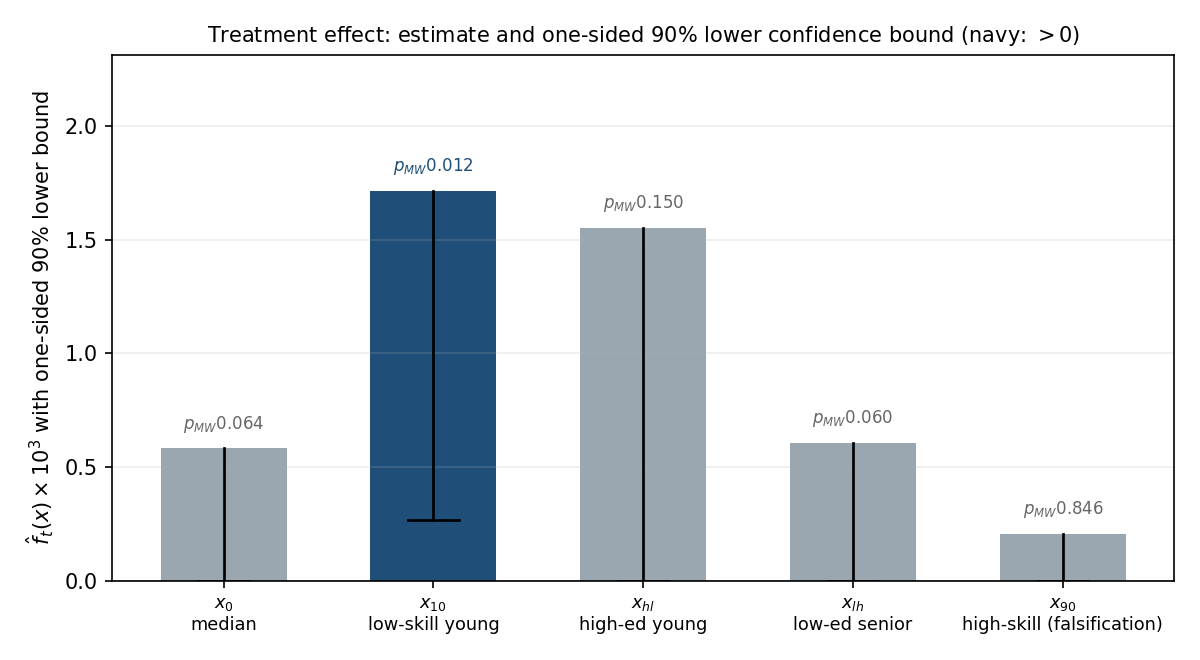}
\caption{Treatment effect $\hat f_t(x,\calY)$ at five covariate values with
the one-sided $90\%$ lower confidence bound (navy: $\mathrm{LCB}>0$); the
annotation gives the focused-test $p$-value $\hat p_{\calY_0}$.}
\label{fig:nj_fhat}
\end{figure}

\paragraph{Low-education, low-experience workers ($x_{10}$: educ=10, exper=2).}
The clearest effect is at $x_{10}$: $\hat f_t=1.71\times10^{-3}$, with the
focused test rejecting on the corridor ($p_{\calY_0}=0.012$) and the
full-distribution test rejecting at the $10\%$ level ($T_n=73.0$,
$\hat c_{0.90}=64.8$, $p=0.054$); the one-sided $90\%$ lower bound is
$0.27\times10^{-3}$ (full support) and $0.97\times10^{-3}$ (corridor). The shape
is the canonical minimum-wage spike: $\hat\Delta_t(y)$ falls by up to $0.12$
inside the corridor $[\$4.25,\$5.10]$ and is near zero outside it. Probability
mass previously below \$5.10 is shifted up to just above the new floor. The
supremum is attained within the corridor, so the focused test is the
sharper instrument here.

\paragraph{High-education, low-experience workers ($x_{hl}$: educ=16, exper=2).}
The point estimate $\hat f_t=1.55\times10^{-3}$ is sizeable but \emph{not}
significant ($T_n=37.1$, $p=0.686$; corridor $p_{\calY_0}=0.150$), with a wide
band reflecting the small high-education--low-experience cell. The pointwise confidence intervals of the CDF differences lie largely above zero, which would be consistent with the
wage-compression/ripple pattern of \citet{AutorManningSmith2016}.

\paragraph{Median worker ($x_0$: educ=12, exper=10).}
The median worker shows a small effect that is marginal on the corridor only:
$\hat f_t=0.58\times10^{-3}$ ($T_n=26.7$, $p=0.294$; $p_{\calY_0}=0.064$),
consistent with a modest direct effect for the share of median workers whose
wages were near the old floor.

\paragraph{Low-education, high-experience workers ($x_{lh}$: educ=10, exper=37).}
The estimate is small, $\hat f_t=0.61\times10^{-3}$. The full-distribution
test does not reject ($T_n=28.6$, $p=0.750$); the focused test is only marginal
($p_{\calY_0}=0.060$, rejecting at the $10\%$ level), and the one-sided lower
bound is zero ($\mathrm{LCB}=0$, i.e.\ the interval $[0,\infty)$). There is no
firmly resolved effect near the MW threshold for this group, consistent with
their wages lying well above \$5.10.

\paragraph{High-education, high-experience workers ($x_{90}$: educ=16, exper=37).}
For the high-education/high-experience point, we have no evidence for an effect: $\hat f_t=0.21\times10^{-3}$
with no rejection on either test ($T_n=27.6$, $\hat c_{0.90}=80.2$, $p=0.824$;
corridor $p_{\calY_0}=0.846$). The absence of any corridor effect for a group
whose wages lie far above the floor is as expected and supports the
identifying assumptions.

\paragraph{Overall pattern.}
The DR-SC approach reveals a clear heterogeneity pattern. Under the
April--March design the minimum-wage effect is sharply and \emph{specifically}
concentrated in the MW corridor for low-education/low-experience young workers ($x_{10}$), the group
most directly affected by the policy. For the high-education/high-experience group ($x_{90}$), there is no evidence for an effect, and the remaining groups show at most marginal corridor effects. Relative
to a calendar-year periodisation, which we have also considered, less of the response is attributed to broad
concurrent shifts, consistent with the cleaner pre-trends documented above.

This complements the aggregate DiD
evidence of \citet{CardKrueger1994}: While their design isolates average
employment effects, the DR-SC approach reveals precisely \emph{which} part of
the conditional wage distribution was affected, and for which worker groups.

Regarding the interpretation as causal effect, under PTP, $\hat f_t(x,\calY)$ estimates the total causal effect of all
NJ-specific developments relative to the synthetic control between the pre-
and post-treatment periods. For low-education/low-experience workers, the minimum wage
increase is the most plausible dominant cause. For other groups, the effects
may reflect a combination of the MW and concurrent NJ-specific developments.
The DR-SC method provides a clean characterisation of \emph{where} in the
distribution the effects occur, which is informative about the operative
mechanisms even when causal attribution is uncertain.

\subsection{Robustness: Ridge-Augmented Weights}
A variance decomposition of \eqref{sandwich} shows the pointwise bands are dominated by the
weight-estimation term $V_w$ ($40$--$60\%$ of the pointwise variance across the
five points, against only $\approx15\%$ for the treated cell), reflecting the
ill-conditioned Gram matrix ($\kappa(\hat G)\approx10^5$). As a robustness check
we re-estimate with the ridge weights of Remark~\ref{rem:ridge}, selecting
$\lambda$ by leave-one-period-out cross-validation on the pre-treatment fit.

The chosen $\lambda^*\approx3.8\times10^{-3}$ (of order $n^{-1/2}$) lowers $\kappa$ to
$1.1\times10^4$, evens out the weights (the largest falls from $0.52$ to $0.24$),
and improves the held-out pre-fit by about $20\%$, indicating that the
unregularised weights mildly overfit. The substantive conclusions are unchanged:
The effect stays concentrated in the corridor for $x_{10}$, no evidence for an effect for $x_{90}$, and all pre-trend tests still pass. Regularisation narrows the
pointwise bands by roughly a third and sharpens the $x_{10}$ inference (its
one-sided lower bound moves well above zero). Because $\lambda^*$ introduces a
small, deliberate shrinkage of the weights towards equality (in the spirit of
augmented synthetic control \citep{BenMichaelEtAl2021}), we report it as a
robustness analysis rather than the main specification.

\section{Conclusion}
\label{sec:conclusion}
%% ---------------------------------------------------------------

We have proposed a SC estimator for conditional distribution functions in the
semiparametric DR framework, with three main contributions. First, the Parallel
Trends in Parameters assumption keeps the counterfactual within the model class, and
dropping non-negativity on weights yields a closed-form estimator. Second, both DR
estimation error and weight estimation error contribute at the same $\sqrt{n}$ rate
to the variance of the counterfactual, and both components are characterised
explicitly. Third, a two-stage inference procedure is proposed: a supremum test whose null distribution is approximated by Gaussian process simulation, followed by a plug-in confidence interval for the integrated difference when the null is rejected. The supremum statistic has a valid non-degenerate limit under $H_0$, grows to infinity under $H_1$, and yields $p$-values that fully exploit the large-$n$ precision of the DR estimators.

The New Jersey application illustrates how the proposed framework can uncover heterogeneous distributional patterns across the covariate space. The test statistic $T_n(x,t,\calY_0)$ with point-specific critical values from GP simulation reveals that the estimated effect is concentrated in the minimum-wage corridor for young low-educated workers ($x_{10}$: focused $p_{\calY_0}=0.012$ vs.\ full $p=0.054$), consistent with a direct minimum wage effect, while the other groups show at most marginal corridor effects and the high-education/experience group does not show any noticeable effect ($p_{\calY_0}=0.846$). Pre-trend tests pass on both the full distribution and the corridor. These patterns are invisible in aggregate SC analyses.

Future directions include the theory and applications of uniform confidence bands
for $f(x)$ over a covariate set $\calX$, enabling formal tests for
x-heterogeneity of the treatment effect, and a nonparametric estimation of the link function in the DR model.

\bibliographystyle{apalike}
\bibliography{SC_DR_refs}

\begin{thebibliography}{}

\bibitem[Abadie, 2021]{Abadie2021}
Abadie, A. (2021).
\newblock Using synthetic controls: Feasibility, data requirements, and methodological aspects.
\newblock {\em Journal of Economic Literature}, 59(2):391--425.

\bibitem[Abadie et~al., 2010]{AbadieEtAl2010}
Abadie, A., Diamond, A., and Hainmueller, J. (2010).
\newblock Synthetic control methods for comparative case studies: Estimating the effect of {C}alifornia's tobacco control program.
\newblock {\em Journal of the American Statistical Association}, 105:493--505.

\bibitem[Abadie and Gardeazabal, 2003]{AbadieGardeazabal2003}
Abadie, A. and Gardeazabal, J. (2003).
\newblock The economic costs of conflict: A case study of the {B}asque country.
\newblock {\em American Economic Review}, 93:113--132.

\bibitem[Arkhangelsky et~al., 2021]{ArkhangelskEtAl2021}
Arkhangelsky, D., Athey, S., Hirshberg, D.~A., Imbens, G.~W., and Wager, S. (2021).
\newblock Synthetic difference-in-differences.
\newblock {\em American Economic Review}, 111:4088--4118.

\bibitem[Autor et~al., 2013]{AutorDornHanson2013}
Autor, D.~H., Dorn, D., and Hanson, G.~H. (2013).
\newblock The china syndrome: Local labor market effects of import competition in the united states.
\newblock {\em American Economic Review}, 103(6):2121--2168.

\bibitem[Autor et~al., 2016]{AutorManningSmith2016}
Autor, D.~H., Manning, A., and Smith, C.~L. (2016).
\newblock The contribution of the minimum wage to {US} wage inequality over three decades: A reassessment.
\newblock {\em American Economic Journal: Applied Economics}, 8(1):58--99.

\bibitem[Ben-Michael et~al., 2021]{BenMichaelEtAl2021}
Ben-Michael, E., Feller, A., and Rothstein, J. (2021).
\newblock The augmented synthetic control method.
\newblock {\em Journal of the American Statistical Association}, 116(536):1789--1803.

\bibitem[Biewen and Erhardt, 2025]{BiewenErhardt2025}
Biewen, M. and Erhardt, P. (2025).
\newblock Using post-regularization distribution regression to measure the effects of a minimum wage on hourly wages, hours worked and monthly earnings.
\newblock {\em The Econometrics Journal}, page forthcoming.

\bibitem[Blanchard and Katz, 1992]{BlanchardKatz1992}
Blanchard, O.~J. and Katz, L.~F. (1992).
\newblock Regional evolutions.
\newblock {\em Brookings Papers on Economic Activity}, 1992(1):1--75.

\bibitem[Callaway and Li, 2019]{CallawayLi2019}
Callaway, B. and Li, T. (2019).
\newblock Quantile treatment effects in difference in differences models with panel data.
\newblock {\em Quantitative Economics}, 10:1579--1618.

\bibitem[Card and Krueger, 1994]{CardKrueger1994}
Card, D. and Krueger, A.~B. (1994).
\newblock Minimum wages and employment: A case study of the fast-food industry in new jersey and pennsylvania.
\newblock {\em American Economic Review}, 84(4):772--793.

\bibitem[Card and Lemieux, 2001]{CardLemieux2001}
Card, D. and Lemieux, T. (2001).
\newblock Can falling supply explain the rising return to college for younger men? a cohort-based analysis.
\newblock {\em Quarterly Journal of Economics}, 116(2):705--746.

\bibitem[Chen, 2023]{Chen2023}
Chen, J. (2023).
\newblock Synthetic control as online linear regression.
\newblock {\em Econometrica}, 91(2):465--491.

\bibitem[Chen and Feng, 2026]{ChenFeng2026}
Chen, S. and Feng, J. (2026).
\newblock Group-heterogeneous changes-in-changes and distributional synthetic controls.
\newblock {\em arXiv}, 2307.15313v2.

\bibitem[Chernozhukov et~al., 2010]{Chernozhukov2010}
Chernozhukov, V., Fern{\'a}ndez-Val, I., and Galichon, A. (2010).
\newblock Quantile and probability curves without crossing.
\newblock {\em Econometrica}, 78(3):1093--1125.

\bibitem[Chernozhukov et~al., 2013]{ChernozhFVMelly2013}
Chernozhukov, V., Fern{\'a}ndez-Val, I., and Melly, B. (2013).
\newblock Inference on counterfactual distributions.
\newblock {\em Econometrica}, 81:2205--2268.

\bibitem[Chernozhukov et~al., 2021]{ChernozhukovEtAl2021}
Chernozhukov, V., W{\"u}thrich, K., and Zhu, Y. (2021).
\newblock An exact and robust conformal inference method for counterfactual and synthetic controls.
\newblock {\em Journal of the American Statistical Association}, 116(536):1849--1864.

\bibitem[Dette et~al., 2025]{dette2025}
Dette, H., M{\"o}llenhoff, K., and Wied, D. (2025).
\newblock Practically significant differences between conditional distribution functions.
\newblock {\em arXiv}, 2506.06545.

\bibitem[DiNardo et~al., 1996]{DiNardoFortinLemieux1996}
DiNardo, J., Fortin, N.~M., and Lemieux, T. (1996).
\newblock Labor market institutions and the distribution of wages, 1973--1992: A semiparametric approach.
\newblock {\em Econometrica}, 64(5):1001--1044.

\bibitem[Doudchenko and Imbens, 2016]{DoudchenkoImbens2016}
Doudchenko, N. and Imbens, G.~W. (2016).
\newblock Balancing, regression, difference-in-differences and synthetic control methods: A synthesis.
\newblock NBER Working Paper No.~22791.

\bibitem[Ferman and Pinto, 2021]{FermanPinto2021}
Ferman, B. and Pinto, C. (2021).
\newblock Synthetic controls with imperfect pre-treatment fit.
\newblock {\em Quantitative Economics}, 12(4):1197--1221.

\bibitem[Fernández-Val et~al., 2026]{Fernandezval2026}
Fernández-Val, I., Meier, J., van Vuuren, A., and Vella, F. (2026).
\newblock A simple distributional difference-in-differences estimator for univariate and bivariate outcomes.
\newblock {\em arXiv}, 2409.02311v3.

\bibitem[Firpo et~al., 2009]{FirpoFortinLemieux2009}
Firpo, S., Fortin, N.~M., and Lemieux, T. (2009).
\newblock Unconditional quantile regressions.
\newblock {\em Econometrica}, 77(3):953--973.

\bibitem[Flood et~al., 2025]{flood2025ipums}
Flood, S., King, M., Rodgers, R., Ruggles, S., Warren, J.~R., Backman, D., Breton, E., Cooper, G., Rivera~Drew, J.~A., Richards, S., Van~Riper, D., and Williams, K. C.~W. (2025).
\newblock {IPUMS CPS: Version 13.0 [Dataset]}.

\bibitem[Foresi and Peracchi, 1995]{ForesiPeracchi1995}
Foresi, S. and Peracchi, F. (1995).
\newblock The conditional distribution of excess returns: An empirical analysis.
\newblock {\em Journal of the American Statistical Association}, 90:451--466.

\bibitem[Gunsilius, 2023]{Gunsilius2023}
Gunsilius, F.~F. (2023).
\newblock Distributional synthetic controls.
\newblock {\em Econometrica}, 91:1105--1127.

\bibitem[Katz and Murphy, 1992]{KatzMurphy1992}
Katz, L.~F. and Murphy, K.~M. (1992).
\newblock Changes in relative wages, 1963--1987: Supply and demand factors.
\newblock {\em Quarterly Journal of Economics}, 107(1):35--78.

\bibitem[Klein, 2024]{Klein2024}
Klein, N. (2024).
\newblock Distributional regression for data analysis.
\newblock {\em Annual Review of Statistics and Its Application}, 11:321--346.

\bibitem[Kneib et~al., 2023]{KneibEtAl2023}
Kneib, T., Silbersdorff, A., and S{\"a}fken, B. (2023).
\newblock Rage against the mean---{A} review of distributional regression approaches.
\newblock {\em Econometrics and Statistics}, 26:99--123.

\bibitem[Lechner, 2011]{Lechner2011}
Lechner, M. (2011).
\newblock The estimation of causal effects by difference-in-difference methods.
\newblock {\em Foundations and Trends in Econometrics}, 4(3):165--224.

\bibitem[Machado and Mata, 2005]{MachadoMata2005}
Machado, J. A.~F. and Mata, J. (2005).
\newblock Counterfactual decomposition of changes in wage distributions using quantile regression.
\newblock {\em Journal of Applied Econometrics}, 20(4):445--465.

\bibitem[Melly, 2005]{Melly2005}
Melly, B. (2005).
\newblock Decomposition of differences in distribution using quantile regression.
\newblock {\em Labour Economics}, 12(4):577--590.

\bibitem[Neumark and Wascher, 1992]{NeumarkWascher1992}
Neumark, D. and Wascher, W. (1992).
\newblock Employment effects of minimum and subminimum wages: Panel data on state minimum wage laws.
\newblock {\em ILR Review}, 46(1):55--81.

\bibitem[Rothe and Wied, 2013]{RotheWied2013}
Rothe, C. and Wied, D. (2013).
\newblock Misspecification testing in a class of conditional distributional models.
\newblock {\em Journal of the American Statistical Association}, 108:314--324.

\bibitem[Spady and Stouli, 2025]{spady:2025}
Spady, R. and Stouli, S. (2025).
\newblock Gaussian transforms modeling and the estimation of distributional regression functions.
\newblock {\em Econometrica}, 93(5):1885--1913.

\bibitem[van~der Vaart and Wellner, 1996]{VdVWellner1996}
van~der Vaart, A.~W. and Wellner, J.~A. (1996).
\newblock {\em Weak Convergence and Empirical Processes}.
\newblock Springer, New York.

\bibitem[Wied, 2024]{Wied2024}
Wied, D. (2024).
\newblock Semiparametric distribution regression with instruments and monotonicity.
\newblock {\em Labour Economics}, 90:102565.

\end{thebibliography}

%% ---------------------------------------------------------------
\appendix
\section{Technical Lemmas}
\label{app:lemmas}
%% ---------------------------------------------------------------

The following lemmas underlie the proof of Theorem~\ref{thm:main}. A direct consequence of \citet{ChernozhFVMelly2013}, Theorem~5.2, is

\begin{lemma}[DR convergence]\label{lem:DRconv}
 Under Assumptions~\ref{ass:sampling}--\ref{ass:regularity}, for each $(i,t)$:
\[
\sqrt{n_{it}}(\hat{\theta}_{it}(\cdot)-\theta_{it}(\cdot))\rightsquigarrow\mathbb{G}_{it}(\cdot) \quad \text{in } l^{\infty}(\mathcal{Y})^{p},
\]
where $\mathbb{G}_{it}$ is a zero-mean Gaussian process with covariance $\mathcal{I}_{it}(y)^{-1}\Sigma_{it}(y,y')\mathcal{I}_{it}(y')^{-1}$, with the Fisher information
$\mathcal{I}_{it}(y)=\E\!\left[\lambda(X'\theta_{it}(y))^2\big/\bigl(\Lambda(X'\theta_{it}(y))(1-\Lambda(X'\theta_{it}(y)))\bigr)\,XX'\right]$ and
\begin{eqnarray*}
\Sigma_{it}(y,y')&=& \E\!\left[\frac{\lambda(X'\theta_{it}(y))\,\lambda(X'\theta_{it}(y'))}{\Lambda(X'\theta_{it}(y)) (1-\Lambda(X'\theta_{it}(y)))\,\Lambda(X'\theta_{it}(y'))(1-\Lambda(X'\theta_{it}(y')))} \right. \cdot \\ && \left. \left(\mathbf{1}\{Y\le y\wedge y'\}-F_{it}(y|X)F_{it}(y'|X)\right)XX'\right].
\end{eqnarray*}
The factors $1/(\Lambda(1-\Lambda))$ in $\mathcal{I}_{it}$ and $\lambda\lambda'/(\Lambda(1-\Lambda)\Lambda'(1-\Lambda'))$ in $\Sigma_{it}$ are the standardized scores of the binary log-likelihood~\eqref{eq:DR_estim}; at $y=y'$ the information identity gives $\mathcal{I}_{it}(y)^{-1}\Sigma_{it}(y,y)\mathcal{I}_{it}(y)^{-1}=\mathcal{I}_{it}(y)^{-1}$.
The processes $\mathbb{G}_{it}$ are mutually independent by Assumption~\ref{ass:sampling}. Stacking and scaling by $\sqrt{n}$ yields joint convergence in $l^\infty(\calY_0)^p$: $\sqrt{n}(\hat\theta_{it}-\theta_{it})\wto\mathbb{H}_{it}$ in $l^\infty(\calY_0)^p$ with $\Cov(\mathbb{H}_{it}(y),\mathbb{H}_{it}(y'))=r_{it}^{-1}\,\mathcal{I}_{it}(y)^{-1}\Sigma_{it}(y,y')\mathcal{I}_{it}(y')^{-1}$.
\end{lemma}

\begin{lemma}\label{lemma:gram}
(Gram matrix central limit theorem.) Under Assumptions~\ref{ass:sampling}--\ref{ass:regularity},
$\sqrt{n}(\mathrm{vec}(\hat{G})-\mathrm{vec}(G^{*}),\hat{c}-c^{*})\dto\mathcal{N}(0,\Sigma_{G,c})$,
where $\Sigma_{G,c} = O(T_0^{-1})$ is the asymptotic covariance of the time-averaged sample means
\begin{equation*}
\hat{G}_{kl} = (T_0 m)^{-1} \sum_{t=1}^{T_0} \sum_{l'=1}^{m} \hat{\theta}_{k,t}(y_{l'})' \hat{\theta}_{l,t}(y_{l'}), \quad
\hat{c}_k = (T_0 m)^{-1} \sum_{t=1}^{T_0} \sum_{l'=1}^{m} \hat{\theta}_{1,t}(y_{l'})' \hat{\theta}_{k,t}(y_{l'}).
\end{equation*}
Explicitly, by mutual independence of the cells $(i,t)$ (Assumption~\ref{ass:sampling}),
\begin{equation}\label{eq:SigmaGc}
\Sigma_{G,c}=\frac{1}{T_0^{2}}\sum_{t=1}^{T_0}\sum_{i=1}^{J+1}\frac{1}{r_{i,t}}\,\Phi_{i,t}
=O\!\bigl(T_0^{-1}\bigr),\
\Phi_{i,t}=\frac{1}{m^{2}}\sum_{l'=1}^{m}\sum_{l''=1}^{m}A_{i,t}(y_{l'})\,\Omega_{i,t}(y_{l'},y_{l''})\,A_{i,t}(y_{l''})',
\end{equation}
where $\Omega_{i,t}(y,y')=\mathcal{I}_{i,t}(y)^{-1}\Sigma_{i,t}(y,y')\mathcal{I}_{i,t}(y')^{-1}\in\R^{p\times p}$ is the DR covariance of Lemma~\ref{lem:DRconv} and $A_{i,t}(y)\in\R^{(J^2+J)\times p}$ is the loading \emph{matrix} of $\mathbb{H}_{i,t}(y)$ in the linearisation~\eqref{eq:linGc} below: its row associated with the Gram coordinate $(k,l)$ is $\theta_{l,t}(y)'\mathbf{1}\{i=k\}+\theta_{k,t}(y)'\mathbf{1}\{i=l\}$, and its row associated with the $c_k$ coordinate is $\theta_{1,t}(y)'\mathbf{1}\{i=k\}+\theta_{k,t}(y)'\mathbf{1}\{i=1\}$ (each a $1\times p$ row). Hence $\Phi_{i,t}\in\R^{(J^2+J)\times(J^2+J)}$ and $\Sigma_{G,c}\in\R^{(J^2+J)\times(J^2+J)}$, conformable with $\sqrt n(\vecop(\hat G-G^*),\hat c-c^*)\in\R^{J^2+J}$. The factor $1/r_{i,t}$ is the pre-period analogue of the $1/r_{i,t}$ in~\eqref{eq:cov0}, and the $T_0^{-2}$ prefactor over $T_0$ independent periods yields the $O(T_0^{-1})$ rate.
\end{lemma}
\textit{Proof.} $\hat G$ and $\hat c$ are smooth (bilinear) functionals of the fixed-grid DR estimates. Linearizing and inserting the common-scale processes $\mathbb{H}_{i,t}=r_{i,t}^{-1/2}\mathbb{G}_{i,t}$ of Lemma~\ref{lem:DRconv},
\begin{equation}\label{eq:linGc}
\begin{aligned}
\sqrt{n}(\hat G_{kl}-G^*_{kl})&=\frac{1}{T_0m}\sum_{t=1}^{T_0}\sum_{l'=1}^{m}\bigl[\theta_{k,t}(y_{l'})'\mathbb{H}_{l,t}(y_{l'})+\mathbb{H}_{k,t}(y_{l'})'\theta_{l,t}(y_{l'})\bigr]+o_p(1),\\
\sqrt{n}(\hat c_{k}-c^*_{k})&=\frac{1}{T_0m}\sum_{t=1}^{T_0}\sum_{l'=1}^{m}\bigl[\theta_{1,t}(y_{l'})'\mathbb{H}_{k,t}(y_{l'})+\mathbb{H}_{1,t}(y_{l'})'\theta_{k,t}(y_{l'})\bigr]+o_p(1).
\end{aligned}
\end{equation}
The displayed maps are continuous linear functionals of the independent Gaussian processes $\mathbb{H}_{i,t}$ ($\Cov(\mathbb{H}_{i,t}(y),\mathbb{H}_{i,t}(y'))=r_{i,t}^{-1}\Omega_{i,t}(y,y')$); the functional delta method gives joint asymptotic normality, and collecting the independent cell contributions gives~\eqref{eq:SigmaGc}. \qed

The proof of the following two theorems can be found in Appendix~\ref{app:proofs}.

\begin{theorem}[Asymptotic normality of the weight estimator]
\label{thm:weights}
Under Assumptions~\ref{ass:sampling}--\ref{ass:gram}, $\sqrt{n}(\hat{w}-w^{*})\rightsquigarrow\mathcal{N}(0,V_{w})$ where $V_{w}=J_{w}\Sigma_{G,c} J_{w}^{\prime}$ and $J_{w}$ is the Jacobian of $(G,c)\mapsto w^{*}(G,c)$ at $(G^{*},c^{*})$:
\begin{align}
\label{eq:jacobian_w}
\frac{\partial w^*}{\partial c} &= G^{*-1} - \frac{G^{*-1}\mathbf{1}\mathbf{1}'G^{*-1}}{\mathbf{1}'G^{*-1}\mathbf{1}}, \qquad
\frac{\partial w^*}{\partial \mathrm{vec}(G)} = -\,w^{*\prime}\otimes\frac{\partial w^*}{\partial c}.
\end{align}
\end{theorem}

\section{Proofs}
\label{app:proofs}
%% ---------------------------------------------------------------

\textit{Proof of Theorem~\ref{thm:main}.} \textit{Part (a).} Under Assumption~\ref{ass:span} (perfect pre-treatment balance), decompose
\begin{equation*}
\hatthetazerot{t}(y)-\thetazerot{t}(y)=\underbrace{\sum_{i=2}^{J+1}w_{i}^{*}(\hat{\theta}_{i,t}(y)-\theta_{i,t}(y))}_{=:A_t(y)}+\underbrace{\sum_{i=2}^{J+1}(\hat{w}_{i}-w_{i}^{*})\hat{\theta}_{i,t}(y)}_{=:B_t(y)}.
\end{equation*}
For term $A_t$: by Lemma~\ref{lem:DRconv}, $\sqrt{n}A_t(\cdot)\rightsquigarrow\sum_{i}w_{i}^{*}\mathbb{H}_{i,t}(\cdot)$ in $l^{\infty}(\calY_0)^{p}$. For term $B_t$: since $\hat{\theta}_{i,t}(y)\rightarrow\theta_{i,t}(y)$ uniformly in $y$, we have $\sqrt{n}B_t(y) = \Theta_t(y)' \sqrt{n}(\hat{w}-w^*) + o_p(1)$. Both $A_t$ and $B_t$ are $O_{p}(n^{-1/2})$, so neither is negligible. Crucially, $A_t$ depends on the post-period estimators $\mathbb{H}_{i,t}$ while $\sqrt{n}(\hat{w}-w^{*})$ depends on the pre-period estimators $\mathbb{H}_{i,t}$, $t=1,\ldots,T_0$. By Assumption~\ref{ass:sampling} these are independent, so the cross-covariance between $A_t$ and $B_t$ vanishes. The covariance \eqref{eq:cov0} then follows by direct computation, using $\mathrm{Var}(\sqrt{n}(\hat w - w^*)) \pto V_w$.

\textit{Parts (b) and (c).}
Define the map $\phi_x : l^\infty(\calY_0)^p \times l^\infty(\calY_0)^p \to l^\infty(\calY_0)$ by
$\phi_x(\theta_1, \theta_0)(y) = \Lambda(x'\theta_1(y)) - \Lambda(x'\theta_0(y))$.
We verify Hadamard differentiability of $\phi_x$ at $(\theta_{1,t}, \thetazerot{t})$
tangentially to $C(\calY_0)^p \times C(\calY_0)^p$.
Let $t_n \downarrow 0$ and $(h_n, h_n^0) \to (h, h^0)$ in $l^\infty(\calY_0)^p$. By the mean value theorem,
\[
  \frac{\Lambda(x'(\theta_{1,t} + t_n h_n)(y)) - \Lambda(x'\theta_{1,t}(y))}{t_n}
  = \lambda(x'\theta_{1,t}(y) + s_n(y)\, t_n\, x'h_n(y))\, x'h_n(y)
\]
for some $s_n(y) \in (0,1)$. Since $\lambda$ is uniformly continuous on compact sets,
$\|h_n\|_\infty$ is bounded and $t_n \to 0$, this converges uniformly in $y$ to
$\lambda(x'\theta_{1,t}(y))\, x'h(y)$. The same argument applies to the $\theta_0$-component,
giving the Hadamard derivative
$d\phi_x(h,h^0)(y) = \lambda(x'\theta_{1,t}(y))\,x'h(y) - \lambda(x'\thetazerot{t}(y))\,x'h^0(y)$.
Part~(b) now follows from Part~(a) and the functional delta method
(\citealp{VdVWellner1996}, Theorem~3.9.4) applied to $\phi_x$, using
$\sqrt{n}(\hatthetazerot{t}(\cdot) - \thetazerot{t}(\cdot), \hat\theta_{1,t}(\cdot) - \theta_{1,t}(\cdot)) \rightsquigarrow
(\mathbb{G}^0_t, \mathbb{H}_{1,t})$.

For Part~(c), $f_t(x)$ is treated as the Riemann sum $\frac{1}{m}\sum_l\Delta_t(y_l\mid x)^2$ over the fixed grid $\calY_m$, consistently with $\hat f_t(x)$. The functional $\phi^{(2)}_x : (\theta_1,\theta_0) \mapsto
\frac{1}{m}\sum_l [\phi_x(\theta_1,\theta_0)(y_l)]^2$ is Hadamard differentiable
as a composition $\phi^{(2)}_x = \psi \circ \phi_x$, where
$\psi : g \mapsto \frac{1}{m}\sum_l g(y_l)^2$ is continuously differentiable.
The chain rule (\citealp{VdVWellner1996}, Lemma~3.9.3) and the delta method then give
$\sqrt{n}(\hat f_t(x) - f_t(x)) \dto \mathcal{N}(0,\sigma^2_t(x))$ with
$\sigma^2_t(x)$ as stated. \qed

\textit{Proof of Corollary~\ref{cor:orthogonality}.}
Under Assumption~\ref{ass:ptp}, $\thetazerot{t}(y)=\sum_i w_i^*\theta_{i,t}(y)+r_n(y)$,
so the counterfactual estimator satisfies
$\hatthetazerot{t}(y)\pto\sum_i w_i^*\theta_{i,t}(y)=\thetazerot{t}(y)$.
Decompose $\sqrt{n}(\hat f_t(x)-f_t(x))=\sqrt{n}(\hat f_t(x)-\tilde f_t(x))+\sqrt{n}(\tilde f_t(x)-f_t(x))$,
where $\tilde f_t(x):=\int_{\calY_0}[\Lambda(x'\theta_{1,t})-\Lambda(x'\sum_i w_i^*\theta_{i,t})]^2\,dy$.
The first term $\sqrt{n}(\hat f_t(x)-\tilde f_t(x))\dto\mathcal{N}(0,\sigma^2_t(x))$ by the
same delta-method argument as in Theorem~\ref{thm:main}(c), provided $\tilde f_t(x)>0$
(which follows from $f_t(x)>0$ and $\tilde f_t(x)\to f_t(x)$ under the stated conditions).
For the second term, a second-order Taylor expansion gives \eqref{eq:bias_decomp}.
Under condition (a) or (b), $\sqrt{n}(\tilde f_t(x)-f_t(x))\to 0$, and the result
follows by Slutsky's theorem. \qed

\textit{Proof of Theorem~\ref{thm:weights}.}
Consider the closed-form weight mapping $w^*(G,c) = G^{-1}c - G^{-1}\mathbf{1}\,\nu(G,c)$, with scalar multiplier
\[
  \nu(G,c) = \frac{\mathbf{1}'G^{-1}c - 1}{\mathbf{1}'G^{-1}\mathbf{1}}.
\]
We establish the Jacobian of this mapping at $(G^*,c^*)$, using $w^*=G^{-1}(c-\mathbf{1}\nu)$ and $\mathbf{1}'w^*=1$.

\textbf{Derivative with respect to $c$:}
Since $\partial\nu/\partial c=(\mathbf{1}'G^{-1}\mathbf{1})^{-1}\mathbf{1}'G^{-1}$,
\begin{equation*}
  \frac{\partial w^*}{\partial c} = G^{-1} - G^{-1}\mathbf{1}\,\frac{\partial \nu}{\partial c}
  = G^{-1} - \frac{G^{-1}\mathbf{1}\mathbf{1}'G^{-1}}{\mathbf{1}'G^{-1}\mathbf{1}}.
\end{equation*}
We write $P$ for this matrix, the oblique projection onto $\{u:\mathbf{1}'u=0\}$ along $G^{-1}\mathbf{1}$.

\textbf{Derivative with respect to $G$:}
With $\partial G^{-1}/\partial G_{kl}=-G^{-1}e_ke_l'G^{-1}$ ($e_k$ the $k$-th unit vector) and $G^{-1}(c-\mathbf{1}\nu)=w^*$,
\begin{equation*}
  \frac{\partial w^*}{\partial G_{kl}}
  = -\,G^{-1}e_ke_l'w^* - G^{-1}\mathbf{1}\,\frac{\partial \nu}{\partial G_{kl}} .
\end{equation*}
Differentiating the constraint $\mathbf{1}'w^*=1$ gives $\mathbf{1}'\,\partial w^*/\partial G_{kl}=0$, whence $\partial\nu/\partial G_{kl}=-(\mathbf{1}'G^{-1}\mathbf{1})^{-1}\mathbf{1}'G^{-1}e_ke_l'w^*$. Substituting and using the definition of $P$,
\begin{equation*}
  \frac{\partial w^*}{\partial G_{kl}} = -\,P\,e_ke_l'w^* = -\,w^*_l\,P e_k ,
\end{equation*}
so, stacking the columns (column-major $\mathrm{vec}$),
\[
  \frac{\partial w^*}{\partial \mathrm{vec}(G)} = -\,w^{*\prime}\otimes P
  = -\,w^{*\prime}\otimes\frac{\partial w^*}{\partial c},
\]
which is~\eqref{eq:jacobian_w}.
Collecting terms gives the joint Jacobian matrix $J_w$. By Lemma~\ref{lemma:gram}, $\sqrt{n}(\hat{G} - G^*, \hat{c} - c^*)$ converges weakly to $\mathcal{N}(0, \Sigma_{G,c})$. It follows from the delta method that $\sqrt{n}(\hat{w} - w^*) \dto \mathcal{N}(0, V_w)$ with $V_w = J_w \Sigma_{G,c} J_w'$. \qed

\end{document}